\documentclass[a4paper]{article}

\usepackage{amsthm,amsmath,amsfonts,amssymb,amscd,mathrsfs,latexsym}
\usepackage{enumerate,graphicx,cases,extarrows,indentfirst,tabularx}
\usepackage{slashed}

\usepackage{xcolor}
\usepackage[all,color,2cell]{xy}
\usepackage{xypic} \xyoption{all} 
\usepackage{hyperref} 

\usepackage{mathpazo}
\newcommand{\op}{\operatorname}
\newcommand{\C}{\mathbb{C}}

\newcommand{\Q}{\mathbb{Q}}
\newcommand{\Z}{\mathbb{Z}}

\newcommand{\im}{\op{im}}

\providecommand{\abs}[1]{\left\lvert#1\right\rvert}

\newcommand{\abracket}[1]{\left\langle#1\right\rangle}
\newcommand{\bbracket}[1]{\left[#1\right]}
\newcommand{\fbracket}[1]{\left\{#1\right\}}
\newcommand{\bracket}[1]{\left(#1\right)}

\newcommand{\mc}{\mathcal}

\newcommand{\pa}{\partial}

\newcommand{\dbar}{\bar\pa}

\newcommand{\into}{\hookrightarrow}

\newcommand{\iso}{\cong}

\newcommand{\MC}{Maurer-Cartan }

\newtheorem{theorem}{Theorem}[section]
\newtheorem{proposition}[theorem]{Proposition}
\newtheorem{lemma}[theorem]{Lemma}
\newtheorem{remark}[theorem]{Remark}
\newtheorem{definition}[theorem]{Definition}
\newtheorem{example}[theorem]{Example}
\newtheorem{corollary}[theorem]{Corollary}

\newcommand{\A}{\mathcal{A}}
\newcommand{\K}{\mathcal{K}}
\newcommand{\pat}{\partial}
\newcommand{\mtb}{\mathbb}
\newcommand{\lsm}{\lesssim}
\newcommand{\up}{\Upsilon}
\newcommand{\lge}{\langle}
\newcommand{\rge}{\rangle}
\newcommand{\sD}{\slashed{D}}

\DeclareMathOperator{\PV}{PV}

\DeclareMathOperator{\KK}{K}

\DeclareMathOperator{\Tr}{Tr}

\title{On the $L^2$-Hodge theory of Landau-Ginzburg models}

\author{Si Li, Hao Wen}

\newcommand{\Addresses}{{
  \bigskip
  \footnotesize

  S. Li, \textsc{Department of Mathematical Sciences and Yau Mathematical Sciences Center, Tsinghua University,
    Beijing 100084, China}\par\nopagebreak
  \textit{E-mail address}: \texttt{sili@mail.tsinghua.edu.cn}

  \medskip

  H. Wen, \textsc{Yau Mathematical Sciences Center, Tsinghua University,
    Beijing 100084, China}\par\nopagebreak
  \textit{E-mail address} \texttt{hwen@math.tsinghua.edu.cn}

}}

\begin{document}
\maketitle

\begin{abstract}Let $X$ be a non-compact Calabi-Yau manifold and $f$ be a holomorphic function on $X$ with compact critical
locus. We introduce the notion of $f$-twisted Sobolev spaces for the pair
$(X,f)$ and prove the corresponding Hodge-to-de Rham degeneration property via $L^2$-Hodge theoretical methods when $f$ satisfies an asymptotic condition of strongly ellipticity. This leads to a Frobenius manifold via the Barannikov-Kontsevich construction, unifying the Landau-Ginzburg and Calabi-Yau geometry. Our construction can be viewed as a generalization of K.Saito's higher residue and primitive form theory for isolated singularities. As an application, we construct Frobenius manifolds for orbifold Landau-Ginzburg B-models which admit crepant resolutions. 

\end{abstract}

\setcounter{tocdepth}{2}
\tableofcontents


\section{Introduction}
Frobenius manifolds were introduced by Dubrovin as the fundamental algebraic structure of 2d topological field theories \cite{Du}. The first series of examples arise from K. Saito's study  \cite{Sai2} of primitive period maps associated to the germ of a holomorphic map $f: X\to \C$ with an isolated critical point. The pair $(X, f)$ is termed as the Landau-Ginzburg model, and $f$ is called the superpotential.  In \cite{BK}, Barannikov and Kontsevich gave a systematic construction of Frobenius manifolds on compact Calabi-Yau geometries. This construction illustrates the key role of dGBV algebras behind Frobenius manifolds, and has been vastly generalized to the non-commutative/categorical world \cite{B,KKP}.

In this paper we study the Landau-Ginzburg B-model associated to the triple $(X, \Omega_X, f)$ where
\begin{enumerate}
\item [1)] $X$ is a non-compact complex manifold;
\item [2)] $\Omega_X$ is a holomorphic volume form on $X$ (Calabi-Yau form);
\item [3)] $f:X\to \C$ is a holomorphic function with compact critical set $\text{Crit}(f)$. 
\end{enumerate}

There are two natural dGBV algebras  $(\PV(X), \dbar_f, \pa)$ and $(\PV_c(X), \dbar_f, \pa)$ where
$$
 \PV(X)=\Omega^{0,\bullet}(X, \wedge^\bullet T_X)
$$
is the space of smooth polyvector fields, $
  \PV_c(X)\subset \PV(X)
$ is the subspace of polyvector fields with compact support; $\dbar_f$ is the twist of the $\dbar$-operator by the holomorphic 1-form $df$;
$\pa$ is the divergence operator with respect to the holomorphic volume form $\Omega_X$.  See Section \ref{sec-PV} for details.

When $\text{Crit}(f)$ is compact, the embedding $(\PV_c(X), \dbar_f)\subset (\PV(X), \dbar_f)$ is a quasi-isomorphism, and we can use either of them to study the deformation space. To apply Barannikov-Kontsevich construction to obtain a Frobenius manifold structure on the cohomology $H(\PV(X), \dbar_f)$, we need 
\begin{itemize}
\item[(1)] \textbf{Hodge-to-de Rham degeneration}.  In the current context this is about the $E_1$-degeneration of the spectral sequence computing the $(\dbar_f+u \pa)$-cohomology in terms of the $u$-adic filtration on $\PV_c(X)[[u]]$ or $\PV(X)[[u]]$. Here $u$ is a formal variable.
\item[(2)] \textbf{Trace pairing $\Tr(-,-)$}. This is a linear pairing on polyvector fields compatible with $\dbar_f$ and $\pa$. 
\end{itemize}

In the case when $X$ is a compact K\"{a}hler Calabi-Yau manifold (hence $f=0$), the trace map is given by an honest integration again the Calabi-Yau volume form. The Hodge-to-de Rham degeneration follows from the standard Hodge theory on K\"{a}hler geometry. This is the original set-up of  \cite{BK}. 

In the case when $X$ is affine and $f$ has only an isolated singularity, the Hodge-to-de Rham degeneration holds automatically since $H(\PV(X), \dbar_f)$ is concentrated at degree $0$ (isomorphic to the Milnor ring $\text{Jac}(f)$). We can use $\PV_c(X)$ to define the trace pairing as in the compact Calabi-Yau case. It is shown in \cite{LLS}  that the trace pairing gives a geometric phase of higher residue pairings \cite{Sai3} and the Barannikov-Kontsevich construction reproduces a formal analogue of K.Saito's primitive forms \cite{Sai2}. 

Our main interest in this paper is to generalize the above setting to the case when $X$ is a noncompact Calabi-Yau manifold equipped with a complete K\"{a}hler metric $g$ and a holomorphic function $f$ such that
$$
\text{Crit(f)}= \text{compact}. 
$$
The two dGBV algebras modelled by $\PV_c(X)$ and $\PV(X)$ are not suffice to achieve our goal: 
\begin{itemize}
\item $\PV(X)$ is too big for trace pairing (integration) since $X$ is noncompact;
\item $\PV_c(X)$ is too small for Hodge theory since it does not preserve Hodge decomposition. 
\end{itemize}

In this paper, we describe a third model $\PV_{f,\infty}(X)$ (see Definition \ref{f-adapted-PV}) that enjoys both the integration map and the Hodge decomposition. It arises naturally from the study of a version of $f$-twisted  Sobolev spaces (see Section \ref{sec-f-adapted}).  Briefly speaking, $\PV_{f,\infty}(X)$ consists of polyvectors whose arbitary covariant derivatives multiplied by any powers of the norm  $|\nabla f|$ are still $L^2$ integrable (Definition \ref{defn-twisted-poly}). This allows us to use $L^2$ Hodge theory to establish the Hodge-to-de Rham degeneration property when $f$ satisfies an asymptotic condition of strongly ellipticity at infinity (Definition \ref{condition-T}).

Our main results in this paper are summarized as follows (see Theorem \ref{thm-dGBV-property}, \ref{homotopy_abelian}, \ref{lem-quasi-PV}, \ref{thm-quasi-smooth-formal}, \ref{thm-higher-residue-duality}). 

\begin{theorem}\label{main-thm}Let $(X, \Omega_X)$ be a noncompact Calabi-Yau manifold, $g$ be a complete K\"{a}hler metric on $X$, $f$ be a holomorphic function on $X$ with compact critical set.  Assume $(X, g)$ has a bounded geometry and $f$ is strongly elliptic (Definition \ref{condition-T}). Then 
\begin{enumerate}
\item [(1)] $(\PV_{f,\infty}(X),\dbar_f, \pa)$ forms a dGBV algebra with a trace pairing and Hodge-to-de Rham degeneration holds.  
\item [(2)] We have quasi-isomorphic embeddings of complexes 
$$
  (\PV_c(X), \dbar_f) \subset (\PV_{f,\infty}(X), \dbar_f)\subset (\PV(X), \dbar_f).
$$
In particular, the Hodge-to-de Rham degeneration holds for $(\PV(X), \dbar_f, \pa)$ as well. 
\item [(3)] The trace pairing induces a non-degenerate pairing on the cohomology $H(\PV_{f,\infty}(X), \dbar_f)$. It also induces a sesqui-linear pairing on $H(\PV_{f,\infty}(X)[[u]], \dbar_f+u\pa)$ that generalizes K.Saito's higher residue pairing. 
\end{enumerate}
\end{theorem}

The establish of the above theorem is based on $L^2$-Hodge theoretical methods, which is the main part of the current work. Using harmonic polyvectors, we can construct a splitting map (Proposition \ref{existence-good-basis})
$$
H(\PV_{f,\infty}(X), \dbar_f)\to H(\PV_{f,\infty}(X)[[u]], \dbar_f+u\pa)
$$
that is compatible with the pairing on both sides.  Thanks to Theorem \ref{main-thm}, this leads to a Frobenius manifold structure on the cohomology $H(\PV(X), \dbar_f)$ via Birkhoff factorization method \cite{Sai2, BK, B}.  

There are three main classes of examples that Theorem \ref{main-thm} applies (see Section \ref{example} for details). 
\begin{enumerate}
\item [(a)]$(X,\Omega_X) = (\mtb{C}^n, dz_1\wedge \cdots \wedge dz_n)$, $g$ is the standard flat metric, $f$ is a non-degenerate quasi-homogeneous polynomial. Frobenius manifolds were first constructed in this category by K. Saito \cite{Sai2}, and  developed by M. Saito \cite{SaiM} to arbitrary isolated singularities via the Hodge theory of the Brieskorn lattice. Landau-Ginzburg B-models of this type are mirror to the FJRW theory of counting solutions of Witten's equation \cite{FJR2}. See \cite{HLSW, Li, TF} for a recent exposition on such mirror theorems.  At the categorical level, they are mirror to Fukaya-Seidel categories \cite{Se}.

\item [(b)]$(X,\Omega_X) = ((\mtb{C^*})^n, {dz_1\over z_1}\wedge \cdots \wedge {dz_n\over z_n})$, $g$ is the standard complete metric, $f$ is a convenient non-degenerate Laurent polynomial. Landau-Ginzburg B-models of this type are mirror to Gromov-Witten theory on toric varieties \cite{G1, G2, HV} and variant mirror constructions are known in this literature. The corresponding Frobenius manifolds were studied in \cite{B2, Sab2, DS1,DS2}

\item [(c)] $\pi: X\to \C^n/G$ is a crepant resolution of quotient of $\C^n$ by a finite group $G\subset SU(n)$, $\Omega_X$ is the pull-back of the trivial Calabi-Yau form on $\C^n$. Let $f$ be a $G$-invariant polynomial on $\C^n$ with an isolated singularity at the origin. It pulls back to define a superpotential $\pi^*(f)$ on $X$.   Let $g$ be an ALE K\"{a}hler metric. Then the data $(X, \Omega_X, g, \pi^*(f))$ satisfies the conditions of Theorem \ref{main-thm}. The Landau-Ginzburg B-model of $(X,  \pi^*(f))$ is expected to be equivalent to the orbifold Landau-Ginzburg B-model of $(\C^n, f, G)$ with $G$ being the orbifold group. See \cite{Me, V} for some case studies in the context of matrix factorization categories. 

\end{enumerate}

The Hodge-to-de Rham degeneration of Landau-Ginzburg models for proper  $f$ was conjectured by Barannikov-Kontsevich, who also proposed the $L^2$ approach in a seminar talk. In the algebraic world, the first complete proof for this result in terms of twisted de Rham complex appeared in the work of Sabbah \cite{Sab} and Ogus-Vologodsky \cite{OV}. Stronger degeneracy result appeared in Esnault-Sabbah-Yu \cite{ESY} and Katzarkov-Kontsevich-Pantev \cite{KKP2} for the irregular Hodge filtration. When $f$ is not proper, but defined on smooth affine $X$ satisfying a tameness condition, the degeneracy result is proved in \cite{Sab3, Sab4} for the cohomologically tame case and in \cite{NS} for the M-tame case \footnote{The authors are extremely grateful to Claude Sabbah for explaining various degeneracy results}. Certain generalization to quotient stacks was discussed recently in \cite{HP}. The algebraic approach uses the theory of D-modules/microlocal analysis or the characteristic $p$ methods. Our approach uses analysis and is $L^2$-Hodge theoretic.   The notion of f-twisted Sobolev spaces and corresponding harmonic analysis have their own interests.  The $L^2$-pairing naturally generalizes Saito's higher residue pairing, which was first proposed in the physics literature by Losev \cite{Lo}.

The harmonic analysis of Landau-Ginzburg models has its physics context in N=2 supersymmetric quantum mechanics arising from singularities \cite{CGP, CV}, extending the real case of \cite{W}. Its $L^2$ analysis was first studied in \cite{KL} in which Hodge decompositions are established when $f$ satisfies the tame condition of ellipticity (the $k=2$ part of our strongly elliptic condition \eqref{tame}). This is further developed in \cite{Fan} toward $tt^*$-geometry of \cite{CV}. To obtain the Hodge-to-de Rham degeneration property for the pair of operators $(\dbar_f, \pa)$, we need a bit stronger control on the $\pa$-action and establish a weaker version of a variant of $\pa\dbar_f$-lemma (Lemma \ref{weak-ddbar}). In fact, we feel that our strongly elliptic condition \eqref{tame} could be weakened suitably. Since it applies to main examples of our interest, we will focus on this situation. Our degeneration result and Frobenius manifold construction for crepant resolutions of orbifold Landau-Ginzburg B-models in class (c) of examples above seem new and not established yet from the algebraic method mentioned above.

 Another motivation of the current work is to put Landau-Ginzburg B-models into the framework of quantum Kodaria-Spencer theory \cite{BCOV} along the effective renormalization method developed in \cite{CL}.  The harmonic analysis developed in this paper allows the homological regularization scheme analogous to the compact Calabi-Yau case as presented in \cite{CL}. 

\noindent \textbf{Acknowledgements.} 
The authors would like to thank beneficial discussions with Andrei C\v{a}ld\v{a}raru, Huijun Fan, Dmitry Kaledin, Tony Pantev, Claude Sabbah, and Kyoji Saito. The work of S.Li is partially supported by grant 11801300 of NSFC  and grant Z180003 of Beijing Natural Science Foundation. The work of H. Wen is partially supported by Tsinghua Postdoc Grant 100410032. Part of this work was done when SL was visiting Department of Mathematics and Statistics at Boston University and Perimeter Institute for Theoretical Physics in Jan 2019. SL thanks for their hospitality and provision of excellent working enviroment.

\noindent \textbf{Conventions}
\begin{itemize}
\item We will work with graded vector spaces. For linear operators $A, B$ on a graded vector space, 
$$
   [A, B]:= A B-(-1)^{|A||B|}B A
$$ 
always means the graded commutator. Here $|\cdot|$ is the parity of the operator. 

\item We will frequently use $A \lsm B$ to indicate that there exists $c > 0$ such that $A \leq c\cdot B$.
We will also use  $A \approx B$ to indicate $A \lsm B$ and $B \lsm A$ hold simultaneously.
\end{itemize}

\section{$L^2$ Hodge theory for Landau-Ginzburg model}

\subsection{Differential forms and twisted operators}\label{sec-forms}

Let $X$ be a non-compact complex $n$-manifold equipped with a holomorphic volume form $\Omega_X$. Let $f$ be a holomorphic function on $X$, with the set 
$$
\text{Crit}(f)=\{p\in X| df(p)=0\}
$$ 
of critical points being compact. Let $\A^{i,j}(X)$ denote the space of smooth $(i,j)$-forms on $X$ and 
$$
\A(X)=\bigoplus_{i,j}\A^{i,j}(X).
$$ 
The de Rham differential decomposes $d=\pat+\bar\pat$ where 
$$
\pat: \A^{i,j}(X)\to \A^{i+1,j}(X), \quad \bar\pat: \A^{i,j}(X)\to \A^{i,j+1}(X). 
$$ 

Using $f$,  we can define a twisted Cauchy-Riemann operator on $\A(X)$ 
\begin{equation*}
\bar\pat_f := \bar\pat + df \wedge.
\end{equation*}
Clearly we have
\begin{align*}
\bar\pat_f^2 = 0
\quad \text{and} \quad
\pat \bar\pat_f + \bar\pat_f \pat = 0.
\end{align*}
Let $u$ be a formal variable. The above two equations are equivalent to 
$$
Q_f^2 = 0, \quad \text{where}\quad Q_f := \bar\pat_f + u\pat. 
$$
As a result, we obtain two complexes 
$$
(\A(X)[[u]], Q_f), \quad (\A(X)((u)), Q_f)
$$
where $\A(X)[[u]]$ and $\A(X)((u))$ denote respectively the space of formal power series and formal Laurent series in $u$ with coefficients in $\A(X)$.
Similar notations will be used throughout this paper.

\subsection{$L^2$ theory preliminaries}

In this section, we discuss aspects of $L^2$ theory needed in subsequent sections. We assume $X$ is equipped with a complete K\"ahler metric $g$. It defines Hermitian inner products on all tensors bundles. 

We give here explicitly the inner products for differential forms to set up our notations.
In local coordinates,  $g = \sum g_{i\bar j} (dz^i \otimes dz^{\bar j} + dz^{\bar j} \otimes dz^i)$, and 
any $\varphi \in \A^{p,q}(X)$ is  expressed by 
\begin{align*}
\varphi = \frac{1}{p!q!} {\sum_{\substack {i_1, \cdots, i_p, \\ j_1, \cdots, j_q}}} \varphi_{i_1 \cdots i_p, \bar j_1 \cdots \bar j_q} dz^{i_1}\wedge \cdots dz^{i_p} \wedge dz^{\bar j_1} \wedge \cdots \wedge dz^{\bar j_q},
\end{align*}
where $\varphi_{i_1 \cdots i_p, \bar j_1 \cdots \bar j_q}$ is antisymmetric for the $i$-indices and ${\bar j}$-indices. For $\varphi, \psi\in \A^{p,q}(X)$,  one defines their pointwise inner product
\begin{equation*}
(\varphi, \psi)_{\A} := \frac{1}{p!q!} \sum_{i,j,k,l} g^{\bar k_1 i_1} \cdots g^{\bar j_1l_1} \cdots g^{\bar j_q l_q} \varphi_{i_1 \cdots i_p, \bar j_1 \cdots \bar j_q} \overline{\psi_{k_1 \cdots k_p, \bar l_1 \cdots \bar l_q}},
\end{equation*}
where $(g^{\bar i j})$ is the inverse matrix of $(g_{i \bar j})$.
The subscript $\A$ means we are working with differnetial forms, which we will sometimes omit when there is no confusion.
The $L^2$ inner product is defined to be
\begin{equation*}
\lge \varphi,\psi \rge_{\A} := \int_X (\varphi, \psi)_{\A}dv_g,
\end{equation*}
where $dv_g$ is the volume on $X$ induced by $g$.
The $L^2$-norm is denoted by $\big\|\cdot\big\|_{\A}$. 

\begin{definition}
$L^2_{\A}(X)$ is the completion, with respect to the above defined $L^2$-norm, of the subspace of forms in $\A(X)$ that are bounded with respect to the same norm.
The $L^2$ space for all tensor bundles are defined similarly.
For tensor bundle $E$, the $L^2$ space of its global sections is denoted by $L^2_E(X)$ and the norm is denoted by $\big\|\cdot\big\|_E$
\end{definition}

Let $\nabla$ be the unique  $g$-compatible torsion-free connection on the complexified tangent bundle $T_{\mtb{C}}X$. It defines a connection on any tensor bundle $E$
$$
\nabla: E\to T_{\mtb{C}}^*X\otimes E
$$ 
which we still denote by $\nabla$.  Iterated applications of $\nabla$ gives
$$
\nabla^k: E\to (T_{\mtb{C}}^*X)^{\otimes k}\otimes E.
$$

\begin{definition}
The Sobolev spaces $W^{k,2}_E(X)$ are defined as subspaces of $L^2_E(X)$ 
\begin{equation*}
W^{k,2}_E(X) := \fbracket{\phi\in L^2_E(X) | \nabla^i \phi \in L^2_{(T_{\mtb{C}}^*X)^{\otimes i}\otimes E}(X) \text{ for all }  0 \leq i \leq k}.
\end{equation*}
The $W^{k,2}_E$-norm is defined by
\footnote{This is equivalent to the usual definition by mean inequality.}
\begin{equation*}
\big\|\phi\big\|_{W^{k,2}_E} := \sum_{i=0}^k \big\|\nabla^i \phi\big\|_{(T_{\mtb{C}}^*X)^{\otimes i}\otimes E}.
\end{equation*}
\end{definition}

Recall that a complete Riemannian manifold $(X, g)$ is said to have a \emph{bounded geometry} if it has positive injective radius and has a uniform bound for each order of covariant derivative of the Riemannian curvature tensor.
Sobolev spaces are well behaved on Riemannian manifolds with bounded geometry (see \cite{Heb}).
In particular, the Density theorem and Sobolev's embedding theorem hold.
Throughout this paper, we work with $(X, g)$ which has bounded geometry.

The Cauchy-Riemann operator $\bar\pat$, initially defined on $\A_c(X)$, is densely defined on $L^2_{\A}(X)$.
It has a natural closed extension, still denoted by $\bar\pat$, with domain
\begin{equation*}
\text{Dom}(\bar\pat) := \{\varphi \in L^2_{\A}(X)|\bar\pat \varphi \in L^2_{\A}(X)\}.
\end{equation*}
Its adjoint $\bar\pat^*$ is defined through the equality
\begin{equation*}
\lge \bar\pat^* \varphi,\psi \rge_{\A}
= \lge \varphi,\bar\pat \psi \rge_{\A}, \quad \forall \psi \in \text{Dom}(\bar\pat).
\end{equation*}

Let $\sD := \bar\pat + \bar\pat^*$ be the Dirac-type operator on $L^2_{\A}(X)$, which is an elliptic self-adjoint operator.
We have the following lemma on equivalence of norms.

\begin{lemma} \label{usual_equiv}
Assume $(X,g)$ has bounded geometry, then the $W^{k,2}_{\A}$-norm is equivalent to the norm $\sum_{i=0}^k \big\|\sD^i \cdot\big\|_{\A}$.
Hence $W^{k,2}_{\A}(X)$ can be represented as
\begin{equation*}
W^{k,2}_{\A}(X) = \{\phi | \sD^i \phi \in L^2_{\A}(X) \text{ for all } 0 \leq i \leq k\}.
\end{equation*}
\end{lemma}

\begin{proof}
This Lemma is a special case of Theorem 3.5 in \cite{Sal}. 
\end{proof}

The twisted Cauchy-Riemann operator $\bar\pat_f$ has also a closed extension on $L^2_{\A}(X)$, with the domain
\begin{equation*}
\text{Dom}(\bar\pat_f) := \{\varphi \in L^2_{\A}(X)|\bar\pat_f \varphi \in L^2_{\A}(X)\}.
\end{equation*}
Its adjoint $\bar\pat_f^*$ is defined through the equality
\begin{equation*}
\lge \bar\pat_f^* \varphi,\psi \rge_{\A}
= \lge \varphi,\bar\pat_f \psi \rge_{\A}, \quad \forall \psi \in \text{Dom}(\bar\pat_f).
\end{equation*}

\begin{definition}
The twisted Laplacian $\Delta_f$ associated to $f$ is defined to be
\begin{equation*}
\Delta_f := [\bar\pat_f, \bar\pat_f^*] = \bar\pat_f \bar\pat_f^*+\bar\pat_f^* \bar\pat_f,
\end{equation*}
whose domain $\text{Dom}(\Delta_f)$ is given by
\begin{equation*}
\text{Dom}(\Delta_f)=\{\varphi \in \text{Dom}(\bar{\pat}_f) \cap \text{Dom}(\bar{\pat}_{f}^*)| \bar{\pat}_f(\varphi) \in \text{Dom}(\bar{\pat}_{f}^{*}), \bar{\pat}_{f}^{*}(\varphi) \in \text{Dom}(\bar{\pat}_f)\}.
\end{equation*}
\end{definition}

In physics literature, the twisted Laplacian $\Delta_f$ represents the Hamiltonian operator of N=2 supersymmetric quantum mechanics arising from singularities \cite{CGP, CV}, extending the real case of \cite{W}. Its harmonic analysis was first studied in \cite{KL}.

The following proposition is straight-forward but fundamental.
\begin{proposition}
The Laplacian $\Delta_f$ is a densely defined, non-negative, linear self-adjoint operator on $L^2_{\A}(X)$.
\end{proposition}

Let us express $\pat$ and $\bar\pat$ in terms of covariant derivatives induced by the K\"ahler metric $g$. Locally
$$
\pat = dz^k \wedge \nabla_k, \quad \bar\pat = d\bar z^{\bar k} \wedge \nabla_{\bar k}  
$$ 
on differential forms. Their adjoints are given by standard expressions
$$
\pat^* = - g^{\bar j i} \nabla_{\bar j} \iota_{\pat_{i}}, \quad \bar \pat^* = - g^{\bar j i} \nabla_i \iota_{\pat_{\bar j}}. 
$$
Here $\iota_{\pat_{i}}$ and $\iota_{\pat_{\bar j}}$  are contractions with ${\pat\over \pat z^i}$ and ${\pat \bar z^{\bar j}}$ respectively.
Bochner-Weitzenb\"ock formula gives
\begin{align*}
\Delta_{\bar\pat} :=\bar \pat \bar\pat^*+\bar\pat^*\bar\pat= -\sum_{\mu,\nu} g^{\bar \nu \mu} \nabla_{\mu} \nabla_{\bar \nu} + Ric,
\end{align*}
where $Ric$ means an algebraic action by Ricci curvature.

In the twisted case by $f$, $\Delta_f$ can be expanded as
\begin{align*}
\Delta_f
=& [\bar\pat, \bar\pat^*] + [\pat f \wedge, \bar\pat^*]+ [\bar\pat, (\pat f \wedge)^*]  + [\pat f \wedge, (\pat f \wedge)^*] \\
=& \Delta_{\bar\pat} + (dz^k \wedge g^{\bar i j} \nabla_j f_k \iota_{\pat_{\bar i}} + d\bar z^{\bar k}\wedge {g^{\bar j i} \nabla_{\bar j} \bar f_{\bar k} \iota_{\pat_{i}} }) + g^{\bar j i} f_i \overline{f_j} \\
=& -\sum_{\mu,\nu} g^{\bar \nu \mu} \nabla_{\mu} \nabla_{\bar \nu} + Ric + L_f + |\nabla f|^2.
\end{align*}
In the above expansion,  
$$
L_f:=dz^k \wedge g^{\bar i j} \nabla_j f_k \iota_{\pat_{\bar i}} + d\bar z^{\bar k}\wedge {g^{\bar j i} \nabla_{\bar j} \bar f_{\bar k} \iota_{\pat_{i}} }
$$
which is a differential operator of order $0$.  $\Delta_f$ is an elliptic operator of order $2$.

\subsection{Hodge theory and harmonics}

We will be mostly interested in $\text{Ker}(\Delta_{f})$, i.e. harmonic forms. For $\Delta_{f}$ to have well behaved spectrum, one should put some restriction on the triple $(X, g, f)$.

\begin{definition} \label{condition-T}
Let $(X,g)$ be a bounded geometry. A holomorphic function $f$ on $X$ is said to be \emph{strongly elliptic} if  for $\forall \epsilon >0, k \geq 2$,
\begin{equation} \label{tame}
\epsilon |\nabla f(z)|^k - |\nabla^k f(z)| \rightarrow +\infty \text{ as } z \rightarrow \infty.  \tag{T}
\end{equation}
Here $z\to \infty$ means $d(z, z_0)\to \infty$, where $z_0\in X$ is any chosen fixed point and $d(z,z_0)$ is the distance between $z$ and $z_0$. This notation will be understood in the same manner throughout this paper. 
\end{definition}

 If we only require the weaker condition \eqref{tame} for $k=2$ only, then it goes back to the notion of elliptic in \cite{KL} which is used to established well-behaved spectrums and harmonics.  There are other related notions in the literature.  There is the tame condition in \cite{Br} meaning the gradient of the function is bounded away from $0$ outside of a compact set. In the work of Sabbah \cite{Sab1} and Hertling \cite{He}, the notion of $M$-tame is frequently used, which ensures global Milnor fibrations in big balls in $\mathbb{C}^{n+1}$. A notion of strongly tame is used later in the work \cite{Fan}.  However, the analysis of  \cite{KL} and \cite{Fan} are not enough to obtain the Hodge-to-de Rham degeneration property for the pair of operators $(\dbar_f, \pa)$. Our strongly elliptic assumption is stronger than the above conditions, but on the other hand leads to general results for most interesting cases as we will see.

The following theorem only require the $k=2$ part of strongly ellipticity. We present it in our context for later use. 

\begin{theorem}[\cite{KL, Fan}] Let $(X, g)$ be a bounded geometry, and $f$ be a holomorphic function which is strongly elliptic. 
Then the twisted Laplacian $\Delta_f$ has purely discrete spectrum.
\end{theorem}
\begin{proof} We sketch a proof here for completeness. By assumption, $(X,g)$ has bounded geometry and $|\nabla f|^2$ is bounded from below and tends to infinity as $z$ goes to infinity. The Shr\"odinger operator $H_0 := -\sum_{\mu,\nu} g^{\bar \nu \mu} \nabla_{\mu} \nabla_{\bar \nu}  -|\nabla f|^2$ has purely discrete spectrum by \cite{KoSh}.
By the strongly elliptic condition \eqref{tame} ($k=2$ part), $H_1 := Ric + L_f$ can be viewed as a compact perturbation to $H_0$, hence $\Delta_f = H_0 + H_1$ also has purely discrete spectrum.
\end{proof}

According to the above theorem, $\text{Ker}(\Delta_f)$ is finite dimensional and we have the following \emph{orthogonal} Hodge decomposition:
\begin{equation*}
L^2_{\A}(X) = \text{Ker}(\Delta_f) \oplus \text{Im}(\bar\pat_f) \oplus \text{Im}(\bar\pat_f^*).
\end{equation*}
In the sequel, we will sometimes use the notation 
\begin{equation*}
\mathcal{H}_{\A} := \text{Ker}(\Delta_f).
\end{equation*}
Elements of $\mathcal{H}_{\A}$ will be called $\Delta_f$-harmonic forms.
As operators, we have
\begin{equation*}
id = P_{\A} + \bar\pat_f \bar\pat_f^* G_{\A} + \bar\pat_f^* \bar\pat_f G_{\A},
\end{equation*}
where $id$ is the identity operator, $P_{\A}$ is the orthogonal projection onto harmonics $\mathcal{H}_{\A}$ and $G_{\A}$ is the corresponding Green's operator. $G_{\A}$ is a compact linear operator which commutes with $\bar\pat_f$ and $\bar\pat_f^*$.

The next theorem is a direct consequence of the ellipticity of $\Delta_f$. 

\begin{theorem}
Tthe following regularity result holds
\begin{enumerate}
\item [1)] $\Delta_f$-harmonic forms are all smooth on $X$, i.e., $\mathcal{H}_{\A}\in \A(X)$;
\item [2)] $G_{\A}$ maps forms in $\A(X)$ to $\A(X)$.
\end{enumerate}
\end{theorem}

One important aspect of the $L^2$ theory on K\"ahler manifolds is the K\"ahler-Hodge identities. This extends to our twisted case as follows. We define another twisted operator $\pat_f := \pat + d \bar f \wedge$ and one can establish a parallel $L^2$ theory. Let
\begin{align*}
L := \omega \wedge = i g_{\mu \bar \nu} dz^{\mu} \wedge dz^{\bar \nu}
\quad
\Lambda := L^* = i g^{\bar \mu \nu} \iota_{\pat_\nu} \iota_{\pat_{\bar \mu}}.
\end{align*}

The next version of K\"ahler-Hodge identities is well-known in the literature. In physics,  they represent N=2 supersymmetry arising from Landau-Ginzburg models  \cite{CGP, CV}.

\begin{proposition}
The following generalized K\"ahler-Hodge identities hold:
\begin{align*}
[\pat_f, \Lambda] =& -i \bar\pat_f^*,\quad  [\bar\pat_f, \Lambda] = i \pat_f^*, \\
[\pat_f^*, L] =& -i \bar\pat_f,\quad [\bar\pat_f^*, L] = i \pat_f.
\end{align*}
\end{proposition}
As an easy consequence, we also have 
$$
[\bar\pat_f, \pat_f^*] = [\pat_f, \bar\pat_f^*] = 0.
$$

Define  $d_{2\text{Re} f} := \bar\pat_f + \pat_f = d + d(f+\bar f) \wedge$, then we have $d_{2\text{Re} f}^* = \bar\pat_f^* + \pat_f^*$. Let 
$$
\Delta_{2\text{Re} f} := [d_{2\text{Re} f}, d_{2\text{Re} f}^*]=d_{2\text{Re} f}d_{2\text{Re} f}^*+d_{2\text{Re} f}^*d_{2\text{Re} f}.
$$ 
As a consequence of K\"ahler-Hodge identities, we have
\begin{corollary}
\begin{align*}
\Delta_f = [\pat_f, \pat_f^*] = \frac{1}{2} \Delta_{2 \text{Re} f}.
\end{align*}
\end{corollary}

Note that we can establish \emph{orthogonal} Hodge decomposition for the operator $\pat_f$, while this corollary implies the decomposition takes the following form:
\begin{equation*}
L^2_{\A}(X) = \text{Ker}(\Delta_f) \oplus \text{Im}(\pat_f) \oplus \text{Im}(\pat_f^*).
\end{equation*}

Let $*$ be the Hodge-star operator on $\A(X)$. It sends $(p,q)$-forms to $(n-q, n-p)$-forms by the relation
\begin{equation*}
(\varphi, \psi) dv = \varphi \wedge * \bar \psi.
\end{equation*}
Here $dv$ is the volume form. The Hodge $*$-operator has the following basic properties: 
\begin{enumerate}
\item[1)] $*$ is real, that is, $\overline{*\alpha} = *\bar \alpha$ for any $\alpha \in \A(X)$;
\item[2)] For $\alpha \in \A^{p,q}(X)$, $*^2(\alpha) = (-1)^{p+q}\alpha$.
\end{enumerate}
The next lemma generalizes the standard identities to our twisted case. 
\begin{lemma}
\begin{equation*}
\bar\pat_f^* = -* \pat_{-f} *, \quad \pat_f^* = -* \bar\pat_{-f} *.
\end{equation*}
\end{lemma}
\begin{proof} For any $\alpha, \beta\in \A_c(X)$,   
\begin{align*}
\abracket{\bar\pat_f \alpha, \beta}=&\int_X  \bar\pat_f \alpha \wedge * \bar\beta
= - \int_X (-1)^{|\alpha|} \alpha \wedge \bar\pat_{-f}(* \bar\beta) \\
=& - \int_X \alpha \wedge ** \bar\pat_{-f}(* \bar\beta) \\
=&- \int_X \alpha \wedge *\overline{* \pat_{-f}* \beta}=-\abracket{\alpha, * \pat_{-f} * \beta}
\end{align*}
Therefore
$
\bar\pat_f^* = -* \pat_{-f} *
$. Its complex conjugate gives
$
\pat_f^* = -* \bar\pat_{-f} *.
$
\end{proof}
\begin{corollary} \label{*_transf}
Let $\alpha$ be a smooth differential $k$-form, then 
\begin{align*}
\alpha \in \text{Ker}(\bar\pat_f) &\iff *\bar\alpha \in \text{Ker}(\bar\pat_{-f}^*)\\
\alpha \in \text{Ker}(\pat_f) & \iff *\bar\alpha \in \text{Ker}(\pat_{-f}^*)\\
\alpha \in \text{Im}(\bar\pat_f) &\iff *\bar\alpha \in \text{Im}(\bar\pat_{-f}^*)\\
\alpha \in \text{Im}(\pat_f) &\iff *\bar\alpha \in \text{Im}(\pat_{-f}^*)
\end{align*}
\end{corollary}
\begin{proof}
\begin{align*}
\bar\pat_f \alpha = 0
\iff \pat_f \bar\alpha = 0 
\iff * \pat_f ** \bar\alpha = 0 
\iff -\bar\pat_{-f}^* * \bar\alpha = 0.
\end{align*}
\begin{align*}
\alpha = \bar\pat_f \beta
\iff *\bar\alpha = *\pat_f \bar\beta 
\iff *\bar\alpha = \pm *\pat_f ** \bar\beta 
\iff *\bar\alpha = \mp \bar\pat_{-f}^* * \bar\beta.
\end{align*}
This proves half of the corollary. The rest can be obtained by complex conjugate.
\end{proof}

\begin{example}[$A_1$-singularity] \label{A_1_example}
Let $X = \mtb{C}$, $g$ the standard flat metric and $f = \frac{1}{2} z^2$ be the Landau-Ginzburg model of $A_1$ singularity.
Let $h := 2\text{Re}f = x^2-y^2$, then we have
\begin{align*}
\Delta_f = \frac{1}{2} \Delta_h := \frac{1}{2} [d + dh \wedge, d^* + (dh \wedge)^*].
\end{align*}
Hence we have the relation of spectrum: $\sigma_{\Delta_f} = \frac{1}{2}\sigma_{\Delta_h}$.
On the other hand, $\Delta_{h}$ is the Laplacian appeared in the Witten deformation \cite{W}.  It is known that the spectra of $\Delta_h$ is purely discrete and
\begin{equation*}
\text{Ker}(\Delta_h)= \mtb{C} \{e^{-2(x^2+y^2)} dy\} = \mtb{C} \{e^{-2|z|^2} (dz-d\bar z)\}.
\end{equation*}
\end{example}

\subsection{Three classes of Landau-Ginzburg models} \label{example}

In this subsection, we discuss three classes of Landau-Ginzburg models  that satisfy the strongly elliptic condition \eqref{tame} in Definition \ref{condition-T}.

\subsubsection{Invertible quasi-homogeneous polynomials}

Let $X = \mtb{C}^n$. We choose the standard metric $g = \frac{1}{2}\sum_i (dz_i \otimes d\bar z_i + d\bar z_i \otimes dz^i)$.
Let $f$ be a quasi-homogeneous polynomial on $\mtb{C}^n$:  $\exists$ rational numbers $q_1, \cdots, q_n\in \Q\cap (0,1)$ such that
\begin{equation*}
f(\lambda^{q_1} z_1, \cdots, \lambda^{q_n} z_n) = \lambda f(z_1, \cdots, z_n)
\qquad \text{ for } \forall \lambda \in \mtb{C}^*.
\end{equation*}
$q_i$ is called the weight of $z_i$ and the total weight of $f$ is $1$. We require $f$ has only an isolated critical point. 

We call $f$ non-degenerate if it has an isolated critical point at the origin and it contains no monomials of the form $z_i z_j, i\neq j$. This implies that weights are uniquely determined and each $q_i$ lies in $\mtb{Q}\cap (0,{1\over 2}]$ \cite{Sai}.  We need the following lemma from \cite{FJR}, who states it for non-degenerate polynomials, but in fact its proof holds for any quasi-homogeneous polynomial with an isolated critical point.

\begin{lemma}[\cite{FJR}, Theorem 5.8]\label{Thm-FJR}
Let $f(z) \in \mtb{C}[z_1, \cdots, z_n]$ be a quasi-homogeneous polynomial with an isolated critical ponit. Assume $z_i$ has weight $q_i \in \Q\cap (0,1)$ and let $\delta_i := \frac{q_i}{\text{min}_j(1-q_j)}$. Then $\exists c>0$ such that
\begin{equation*}
|z_i| \leq c(|\nabla f(z)| + 1)^{\delta_i}, \quad \quad \forall (z_1, \cdots, z_n) \in \mtb{C}^n, \quad 1 \leq i \leq n. 
\end{equation*}
In particular $|\nabla f(z)|\to \infty$ as $z\to \infty$. Note that if $q_i \leq \frac{1}{2}$ for all $i$, then $\delta_i \leq 1$ for all $i$ as well.
\end{lemma}

\begin{theorem}\label{Thm-polynomial}
Assume $X = \mtb{C}^n$, $g = \frac{1}{2}\sum_i (dz_i \otimes d\bar z_i + d\bar z_i \otimes dz^i)$ and $f$ be a quasi-homogeneous polynomial with an isolated critical point and all weights $q_i \leq \frac{1}{2}$. Then $(X, g, f)$ satisfies the strongly elliptic condition \eqref{tame}. 
\end{theorem}

\begin{proof}
Let $\delta_i$ be as in Lemma \ref{Thm-FJR} and $\hat\delta =\min_i \delta_i$. By assumption we know $\hat\delta\leq \delta_i\leq 1$. 

Assume $f = f_1 + f_2 + \cdots + f_n$ is written as a sum of $n$ monomials and $f_l = s_l \prod z_i^{a_{li}}$ for some non-zero constant $s_l$. 
Note that $\pat_{z_{p_1}} \pat_{z_{p_2}} \cdots \pat_{z_{p_k}} f_l$ is quasi-homogenous of weight $1-q_{p_1}-\cdots -q_{p_k}$. Lemma \ref{Thm-FJR} implies the existence of $c_1>0$ such that 
\begin{align*}
\quad& |\nabla_{z_{p_1}} \nabla_{z_{p_2}} \cdots \nabla_{z_{p_k}} f_l| 
= |\pat_{z_{p_1}} \pat_{z_{p_2}} \cdots \pat_{z_{p_k}} f_l| \\
\leq& c_1 (|\nabla f| + 1)^{1-q_{p_1}-\cdots -q_{p_k}\over \text{min}_j(1-q_j)} 
= c_1 (|\nabla f| + 1)^{\frac{1}{\text{min}_j(1-q_j)} - \delta_{q_{p_1}} - \cdots - \delta_{q_{p_k}}} \\
\leq& c_1 (|\nabla f| + 1)^{2 - k\hat\delta}\leq c_1 (|\nabla f| + 1)^{k-\hat \delta}. \quad (k\geq 2) 
\end{align*}

Since $|\nabla f(z)|\to \infty$ as $z\to \infty$, we have for any $c_2 >0$
\begin{align*}
&c_1 (|\nabla f| + 1)^{k-\hat \delta}
\leq 2 c_1 (|\nabla f| + 1)^{k-\hat \delta} - c_2 
\leq {2^{k+1} c_1\over (|\nabla f| + 1)^{\hat \delta}} |\nabla f|^k - c_2  
\end{align*}
holds when $z\to \infty$ is sufficiently large. The theorem follows by combining 
\begin{align*}
|\nabla^k f|
\leq \sum_{l, p_i} |\nabla_{z_{p_1}} \nabla_{z_{p_2}} \cdots \nabla_{z_{p_k}} f_l|.
\end{align*}
\end{proof}



\subsubsection{Crepant resolution of Landau-Ginzburg orbifolds}

Let $f: \mtb{C}^n \rightarrow \mtb{C}$ be a quasi-homogeneous polynomial with an isolated singularity. Let $G$ be a finite group acting linearly on $\C^n$ such that $f$ is $G$-invariant. $f$ descends to define a function on the quotient  
\begin{equation*}
 f: \mtb{C}^n/G \rightarrow \mtb{C}.
\end{equation*}

When $G \subset SL(n ,\mtb{C})$, the quotient $\mtb{C}^n/G$ will admit a global nowhere vanishing holomorphic $n$-form $\Omega = dz_1\wedge \cdots \wedge dz_n$.
In cases when $\mtb{C}^n/G$ has a crepant resolution $\pi: X \rightarrow \mtb{C}^n/G$, $\Omega_X=\pi^* \Omega$ becomes a holomorphic volume form on $X$. $f$ also pulls back to define a holomorphic function $f_X=\pi^*(f)$ on $X$. We are interested in the Landau-Ginzburg model associated to $(X, \Omega_X, f_X)$. 

To incorporate with the metric, we  need the notion of ALE manifolds:
\begin{definition} 
Suppose $G$ is a finite subgroup of $SU(n)$ acting freely on $\mtb{C}^n \setminus 0$, and $\pi: X \rightarrow \mtb{C}^n/G$ be a resolution, with complex structure $J$, and let $g$ be a K\"ahler metric on $X$.
We say that $(X, J, g)$ is an Asymptotically Locally Euclidean (ALE for short) K\"ahler manifold, and that $g$ is an ALE K\"ahler metric, if for some $R>0$,
\begin{equation*}
\nabla^k(\pi_*(g) - g_0) = O(r^{-2n-k}) \text{ on } \{z \in \mtb{C}^n/G | r(z) > R\}, \quad \forall k\geq 0.
\end{equation*}
Here $g_0$ is the Euclidean metric on $\mtb{C}^n$ and $r(z) := (\sum_i |z_i|^2)^{1/2}$ is the radius function on $\mtb{C}^n$.
\end{definition}

In \cite{Joy}, D. Joyce proved that when a subgroup $G$ of $SU(n)$ acts freely on $\mtb{C}^n$ away from the origin and $\pi: X \rightarrow \mtb{C}^n/G$ is a crepant resolution, then there exists Ricci-flat ALE metrics on $X$.
This result fits well in our situation.

\begin{theorem} Assume $\pi:X\to \mtb{C}^n/G$ is a crepant resolution.  $f$ is a $G$-invariant quasi-homogenous polynomial on $\mtb{C}^n$ with no weight greater than $\frac{1}{2}$. Let $f_X=\pi^*(f)$ and $g$ be an arbitrary ALE k\"ahler metric on $X$. Then $(X, f_X, g)$ satisfies the strongly elliptic condition \eqref{tame}.
\end{theorem}
\begin{proof}
Since injective radius is Lipschitz continuous and the curvature involves only $g$, its inverse and derivatives, it follows that ALE metric implies bounded geometry.
We need only to verify strongly ellipticity (\ref{tame}) near infinity, which follows from the definition of ALE k\"ahler metric and Theorem \ref{Thm-polynomial}. 
\end{proof}

\subsubsection{Convenient Laurent polynomial on $(\mtb{C}^*)^n$}

Let $X = \mtb({C}^*)^n$ be the complex torus, with complete metric $g := \frac{1}{2} \sum_i (\frac{dz_i}{z_i} \otimes \frac{d\bar z_i}{\bar z_i} + \frac{d\bar z_i}{\bar z_i} \otimes \frac{dz_i}{z_i})$. Let $f: X \rightarrow \mtb{C}$ be a Laurent polynomial of the form:
\begin{equation*}
f(z_1, \cdots, z_n) := \sum_{\alpha \in \mtb{Z}^n} c_{\alpha} z^{\alpha} = \sum_{\alpha \in \mtb{Z}^n} c_{\alpha} z_1^{\alpha_1} \cdots z_n^{\alpha_n}
\end{equation*}
where $\alpha := (\alpha_1, \cdots, \alpha_n)$ is a multi-index.
For every Laurent polynomial $f$, we can define its Newton polytope $\triangle(f)$ by the convex hull in $\mtb{Z}^n$ of the set $\{\alpha| c_{\alpha} \neq 0\}$.
We will say $f$ is convenient if $0 \in \mtb{Z}^n$ lies in the interior of $\triangle(f)$.

Let $\triangle^\prime$ be a face of any dimension of $\triangle(f)$. We denote
\begin{equation*}
f^{\triangle^\prime} := \sum_{\alpha \in \triangle^\prime} c_{\alpha} z^{\alpha}.
\end{equation*}
Define the logarithmic derivative of $f$ with respect to $z_i$ by
\begin{equation*}
f_i(z) := z_i \frac{\pat}{\pat z_i} f(z) = \frac{\pat}{\pat (\log z_i)} f(z).
\end{equation*}
We say $f$ is non-degenerate if for arbitrary face $\triangle^{\prime}$ of $\triangle(f)$, the equations
\begin{equation*}
f^{\triangle^{\prime}}(z) = f_1^{\triangle^{\prime}}(z) = \cdots = f_n^{\triangle^{\prime}}(z)
\end{equation*}
have no common solution on $X$.
The above notion of conveniency and non-degeneracy appeared firstly in \cite{Ko}. The Hodge theory of its Brieskorn lattice were explored in \cite{Sab2, DS1,DS2}. 

\begin{theorem}
Let $X = \mtb({C}^*)^n$, $g = \frac{1}{2} \sum_i (\frac{dz_i}{z_i} \otimes \frac{d\bar z_i}{\bar z_i} + \frac{d\bar z_i}{\bar z_i} \otimes \frac{dz_i}{z_i})$ and $f: X \rightarrow \mtb{C}$ be a convenient non-degenerate Laurent polynomial. Then $(X, g, f)$ satisfies the strongly elliptic condition \eqref{tame}.
\end{theorem}
\begin{proof}
Apply the coordinate change $t_i = \ln z_i, 1 \leq i \leq n$, $t_i\sim t_i+2\pi i$. Then $f = \sum_{\alpha} c_{\alpha} e^{\lge \alpha, t \rge}$, and the metric $g$ becomes the canonical flat metric.
This implies $(X, g)$ has bounded geometry.
Denote $t := (t_1, \cdots, t_n)$, $\text{Re}(t) := (\text{Re}(t_1), \cdots, \text{Re}(t_n))$ and $\text{Im}(t)=(\text{Im}(t_1), \cdots, \text{Im}(t_n))$.

Let us fix an arbitrary point $\hat t$ with $\text{Re}(\hat t)\neq 0$ and consider the ray $l = \mtb{R}_+ \hat t$ . Such a ray approaches infinity of $X$. We first prove $(\ref{tame})$ holds on $l$.  Since $f$ is convenient, 
$$
M:=\max_{\alpha\in \Delta(f)}\fbracket{\frac{\lge \alpha, \text{Re}(\hat t) \rge}{|\text{Re}(\hat t)|}}>0.
$$

Let $\triangle_l$ be the face of $\triangle(f)$ consisting of lattice points achieving the maximum $M$, i.e. 
$$
 \triangle_l=\fbracket{\left.  \alpha\in \triangle(f)\right | \frac{\lge \alpha, \text{Re}(\hat t) \rge}{|\text{Re}(\hat t)|}=M}. 
$$
For $t\in l$, 
\begin{align*}
|\nabla f(e^t)|^k
= \bracket{\sum_i |\sum_{\alpha} c_{\alpha} \alpha_i e^{\lge \alpha, t \rge}|^2 }^{k/2} 
= \bracket{\sum_i \abs{\sum_{\alpha \in \triangle_l} c_{\alpha} \alpha_i 
e^{M |\text{Re}(t)| + i \lge \alpha, \text{Im}(t) \rge} + \sum_{\alpha \notin \triangle_l} c_{\alpha} \alpha_i e^{\lge \alpha, t \rge}}^2}^{k/2}.
\end{align*}

First note that $\sum_{\alpha \in \triangle_l} c_{\alpha} \alpha_i e^{i \lge \alpha, \text{Im}(t) \rge}$ can not vanish for all $i$.
Otherwise, 
$$
f^{\Delta_l}(z)=\sum_{\alpha \in \triangle_l} c_{\alpha} z^{\alpha}, \quad f^{\Delta_l}_i(z)=\sum_{\alpha \in \triangle_l} c_{\alpha} \alpha_i z^{\alpha}
$$
will have a common zero at $z_i=e^{ tg_i}$ since $f^{\Delta_l}(e^t)={1\over M |\text{Re}(t)|}\sum_i \text{Re}(t_i)f^{\Delta_l}_i(e^t)$ while  $f^{\Delta_l}_i(e^{ t})=0$. This contradicts the non-degeneracy of $f$. 

Next since $\sum_{\alpha \in \triangle_l} c_{\alpha} \alpha_i e^{i \lge \alpha, \text{Im}(t) \rge}$ is periodic in each $t_i$,  there exists a constant $c>0$ such that 
\begin{align*}
\sum_i |\sum_{\alpha \in \triangle_l} c_{\alpha} \alpha_i e^{i \lge \alpha, \text{Im}(t) \rge}|^2 > c >0, \quad \forall t\in l. 
\end{align*}
By definition of $\triangle_l$, for $\alpha \notin \triangle_l$, we have $|e^{\lge \alpha, t \rge}| < e^{M|\text{Re}(t)|}$.
Therefore 
\begin{align*}
|\nabla f|^k \geq c_0 e^{kM |\text{Re}(t)|} \quad \text{as } t \text{ tends to } \infty, \quad \text{for some constant}\ c_0>0. 
\end{align*}
On the other hand, for any $\epsilon>0, k\geq 2$,
\begin{align*}
|\nabla^k f|
=& (\sum_{i_1, \cdots, i_k} |\sum_{\alpha} c_{\alpha} \alpha_{i_1} \cdots \alpha_{i_k} e^{\lge \alpha, t \rge}|^2)^{1/2} \nonumber 
\leq C e^{M |\text{Re}(t)|} \nonumber
\leq \epsilon e^{kM |\text{Re}(t)|} \quad \text{as } t \text{ tends to } \infty. 
\end{align*}
Combine the above two inequalities for $ t \to \infty$ , we proved (\ref{tame}) on each ray $l$.

The general case follows from the continuity of $\nabla^k f$, the compactness of the unit sphere in $\mtb{C}^n$, and the standard finite cover argument. 

\end{proof}

\subsection{$f$-twisted Sobolev spaces}\label{sec-f-adapted}

In this section we develop tools toward proving the Hodge-to-de Rham degeneration property for geometric Landau-Ginzburg B-models. 

To do so, we assume $(X,g)$ is a bounded geometry and $f$ is a holomorphic function satisfying the strongly elliptic condition \eqref{tame}. We will construct a suitable subspace of $\A(X) \cap L^2_{\A}(X)$ which carries a differential graded algebra structure and allows the Hodge decomposition. This  puts our Landau-Ginzburg model into the same setting as compact Calabi-Yau case \cite{BK}.

We introduce here the notion of $f$-twisted Sobolev spaces as a generalization of the usual Sobolev spaces but incorporating the twisting by $f$. 

\begin{definition} \label{approp_space}
The spaces $\A_{f,k}(X)$ are subspaces of $L^2_{\A}(X)$ defined as
\begin{align*}
\A_{f,k}(X) := \{\phi | \sD_f^i \phi \in L^2_{\A}(X) \text{ for } \forall \, 0 \leq i \leq k\},
\end{align*}
where $\sD_f := \bar\pat_f + \bar\pat_f^*$
and the $\A_{f,k}$-norm is defined as
\begin{equation*}
\big\|\phi\big\|_{\A_{f,k}} := \sum_{0 \leq i \leq k} \big\|\sD_f^i \phi\big\|_{\A}.
\end{equation*}
$\A_{f,\infty}(X)$ is defined to be the intersection
\begin{equation*}
\A_{f, \infty}(X) := \bigcap\limits_{k \geq 0} \A_{f,k}(X).
\end{equation*}
\end{definition}

\begin{lemma} \label{trivial_op}
$\A_{f,\infty}(X)$ is preserved by $\bar\pat_f, \bar\pat_f^*$ and $G_{\A}$.
\end{lemma}
\begin{proof}
The case for $\bar\pat_f$ and $\bar\pat_f^*$ is obvious.
For $G_{\A}$ the conclusion comes from the commutativity of $G_{\A}$ with $\bar\pat_f, \bar\pat_f^*$ and the fact that $G_{\A}$ is a bounded operator.
\end{proof}

\begin{corollary}\label{cor-Hodge-decomposition}
The Hodge decomposition theorem still holds on $\A_{f,\infty}(X)$, i.e., $\forall \alpha \in \A_{f,\infty}(X)$, one has
\begin{equation*}
\alpha = P_{\A} \alpha + \bar\pat_f \bar\pat_f^* G_{\A} \alpha + \bar\pat_f^* \bar\pat_f G_{\A} \alpha,
\end{equation*}
where $P_{\A} \alpha$, $\bar\pat_f^* G_{\A} \alpha$ and $\bar\pat_f G_{\A} \alpha$ are all in $\A_{f,\infty}(X)$.
\end{corollary}

We give two other different norms which are convenient to use and  equivalent to the $\A_{f,k}$-norm.
\begin{definition}\label{defn-norms}
The first norm $\A'_{f,k}$ is defined as
\begin{align*}
\big\|\varphi\big\|_{\A'_{f,k}} := \sum_{i+j \leq k} \big\||\nabla f|^i \sD^j \varphi\big\|_{\A},
\end{align*}
and we let $\A'_{f,k}(X)=\{\varphi\in L^2_{\A}(X)| \big\|\varphi\big\|_{\A'_{f,k}} <\infty\}$. The second norm $\A_{f,k}''$ is defined as
\begin{align*}
\big\|\varphi\big\|_{\A''_{f,k}} := \sum_{i+j \leq k} \big\||\nabla f|^i \nabla^j \varphi\big\|_{\A}
\end{align*}
and we let $\A''_{f,k}(X)=\{\varphi\in L^2_{\A}(X)| \big\|\varphi\big\|_{\A''_{f,k}} <\infty\}$. 
\end{definition}

It is the the last norm $\A''_{f,k}$ that we will essentially use later in this paper.

\begin{theorem} \label{equi_norm}
The norms $\A_{f,k}$, $\A'_{f,k}$ and $\A''_{f,k}$ are equivalent. Therefore 
$$
\A_{f,k}(X)=\A'_{f,k}(X)=\A''_{f,k}(X).
$$ 
\end{theorem}

We need some preparations to prove this theorem. Firstly, we have the following density theorem.

\begin{lemma}[Density Theorem] \label{density} $\A_c(X)$ is dense in $\A_{f,k}(X)$ with respect to the $\A_{f,k}$-norm. 
\end{lemma}
\begin{proof} $\A_c(X)$ is a dense subspace of $L^2_{\A}(X)$.  $\sD_f$ is a symmetric operator defined on $\A_c(X)$. 
Since $\sD_f$ has the same symbol as $\sD$, Theorem 2.2 of \cite{Che} implies every power of $\sD_f$ is essentially self adjoint.
Thus $\text{Dom}(\overline{\sD_f^i}) = \text{Dom}((\sD_f^i)^*)$, and
\begin{align*}
\overline\A_{c,f,k}(X) := \cap_{i=0}^k \text{Dom}(\overline{\sD_f^i}) = \cap_{i=0}^k \text{Dom}((\sD_f^i)^*) =: \A_{f,k}(X),
\end{align*}
where $\overline\A_{c,f,k}(X)$ is the closure of $\A_c(X)$ in $\A_{f,k}(X)$ under the $\A_{f,k}$-norm.
\end{proof}

Let $E$ be a Hermitian bundle on $X$ with connection.
Together with $g$ and the $g$-compatible torsion free connection $\nabla$, one can define connections, still denoted by $\nabla$, on all the bundles arising from tensors of $TX, T^*X, E, E^*$.
\begin{definition}[\cite{Sal}]\label{defn-Sal}
Define $C^{\infty}_b(E)$ to be the space of smooth sections $s$ of $E$ such that $\nabla^k s$ is bounded for every $k \geq 0$.
The space $\text{Diff}_b^m(E,F)$, $m \geq 0$, is the space of differential operators $P$ from $E$ to $F$ of the form
\begin{align*}
P = \sum_{i=0}^m \xi_i \nabla^i,
\end{align*}
where $\xi_i \in C^{\infty}_b(Hom((T^*X)^{\otimes i} \otimes E, F))$.
\end{definition}

Note that $g \in C^{\infty}_b(T^*X \otimes T^*X)$ and $g^{-1} \in C^{\infty}_b(TX \otimes TX)$.
Under the assumption of bounded geometry, the Riemannian curvature
\begin{align*}
R \in C^{\infty}_b(T^*X \otimes T^*X \otimes End(TX)).
\end{align*}
Moreover, if $P \in \text{Diff}_b^m(E,F), Q \in \text{Diff}_b^n(F,G)$, then $QP \in \text{Diff}_b^{m+n}(E,G)$.

Let $T^*_{\C}X$ denote the complexified tangent bundle.  In the following, we specify to the case
$$
E=\wedge^{\bullet} T^*_{\C}X.
$$ 
It is easy to see that the Dirac type operator  $\sD=\bar \pat+\bar\pat^*$ lies in $\text{Diff}_b^1(E,E)$.

\begin{lemma} \label{easy}
We have bounded inclusions $\A''_{f,k}(X) \hookrightarrow \A'_{f,k}(X) \hookrightarrow \A_{f,k}(X)$.
\end{lemma}
\begin{proof}
We only need to prove $\big\|\varphi\big\|_{\A_{f,k}} \lsm \big\|\varphi\big\|_{\A'_{f,k}} \lsm \big\|\varphi\big\|_{\A''_{f,k}}$ when corresponding norms are defined.

Since  $\sD^i \in \text{Diff}_b^i(E,E)$, we have 
\begin{align*}
|\sD^j \varphi| \lsm \sum_{s=0}^j |\nabla^s \varphi|, \qquad \big ||\nabla f|^i \sD^j \varphi \big | \lsm \sum_{s=0}^j \big ||\nabla f|^i \nabla^s \varphi \big |. 
\end{align*}
It follows that  $\big\|\varphi\big\|_{\A'_{f,k}} \lsm \big\|\varphi\big\|_{\A''_{f,k}}$.

To prove $\big\|\varphi\big\|_{\A_{f,k}} \lsm \big\|\varphi\big\|_{\A'_{f,k}}$, let us write
$$
 \sD_f = \sD+ T_f.
$$
The part $T_f$ depends linearly on $\nabla f, \nabla \bar f$.  Then 
$$
\sD_f^i=(\sD+ T_f)^i=(\sD+ T_f)\cdots (\sD+ T_f)
$$
can be expanded as a sum of operators of the form
\begin{align*}
\text{Ad}_{\sD}^{i_1}(T_f)\cdots \text{Ad}_{\sD}^{i_s}(T_f) \sD^t,
\end{align*}
where $\text{Ad}_{\sD} := [\sD, \cdot]$ and $i_1 + \cdots i_s + t+s = i$.
Strongly elliptic condition $(T)$ together with the fact $\sD \in \text{Diff}_b^1(E,E)$ implies
\begin{align*}
|\text{Ad}_{\sD}^i(T_f)| \lsm |\nabla f|^{i+1} + 1.
\end{align*}
It follows that $\big\|\varphi\big\|_{\A_{f,k}} \lsm \big\|\varphi\big\|_{\A'_{f,k}}$.
\end{proof}

\begin{lemma} \label{hard}
For $\varphi \in \A_c(X)$, we have $\big\|\varphi\big\|_{\A''_{f,k}} \lsm \big\|\varphi\big\|_{\A_{f,k}}$ for all $k\geq 0$.
\end{lemma}
\begin{proof}
Wo prove by induction on $k$.
The $k=0$ case is trivial.
For $k=1$, let $N>0$ be a constant such that $|\nabla^2 f| \leq \frac{1}{2} |\nabla f|^2 + N$. Since $\varphi \in \A_c(X)$, we can use integration by part to estimate
\begin{align*}
\big\|\sD_f \varphi\big\|^2 + 2N \big\|\varphi\big\|^2 =& \lge \sD\varphi, \sD\varphi \rge + \lge L_f \varphi, \varphi \rge + \lge |\nabla f|^2 \varphi, \varphi \rge + 2N \big\|\varphi\big\|^2 \\
\approx& \big\|\sD \varphi\big\|^2 + \big\||\nabla f| \varphi\big\|^2 + \big\|\varphi\big\|^2 \\
\approx& \big\|\nabla \varphi\big\|^2 + \big\||\nabla f| \varphi\big\|^2 + \big\|\varphi\big\|^2.
\end{align*}
Here in the last line we have used the $k=1$ case of Lemma \ref{usual_equiv}. This proves the Lemma for $k=1$. 

Note that in the proof of Theorem 3.5 in \cite{Sal}, the author actually used the Bochner-Weizenb\"ock formula.
As the same formula holds on tensor-valued differential forms, the conclusion of the above $k=1$ case holds also when $\varphi$ is a section of $(T^*_{\mtb{C}}X)^{\otimes k} \otimes E$.

Assume the Lemma holds for $k \leq m$ and we consider the $k=m+1$ case. First we claim that
\begin{equation}\label{estimate-1}
  \big\| \sD_f |\nabla f|^i \nabla^j \varphi\big\|\lsm  \sum_{0\leq s\leq j} \big\| |\nabla f|^{i+s+1} \nabla^{j-s} \varphi\big\|+ \big\| \varphi \big\|_{\A^{''}_{f,i+j}}+ \big\|  |\nabla f|^i \nabla^j \sD_f\varphi\big\|   \tag{$\dagger$}
\end{equation}
and 
\begin{equation}\label{estimate-2}
  \big\| |\nabla f|^i \nabla^{j+1} \varphi\big\|\lsm  \sum_{0\leq s\leq j} \big\| |\nabla f|^{i+s+1} \nabla^{j-s} \varphi\big\|+ \big\| \varphi \big\|_{\A^{''}_{f,i+j}}+ \big\|  \nabla \bracket{|\nabla f|^i \nabla^j \varphi}\big\|   \tag{$\dagger\dagger$}
\end{equation}
Assume \eqref{estimate-1} \eqref{estimate-2} first. We have
\begin{align*}
\big\| |\nabla f|^{m-s} \nabla^{s+1}\varphi\big\|
\stackrel{\eqref{estimate-2}}{\lsm}& \sum_{i\leq s} \big\| |\nabla f|^{m+1-i} \nabla^{i} \varphi\big\| +  \big\| \varphi \big\|_{\A^{''}_{f,m}}+\big\| \nabla \bracket{|\nabla f|^{m-s} \nabla^{s}\varphi}\big\|\\
\stackrel{k=1\ \text{case}}{\lsm}&  \sum_{i\leq s} \big\| |\nabla f|^{m+1-i} \nabla^{i} \varphi\big\| +  \big\| \varphi \big\|_{\A^{''}_{f,m}}+\big\| \sD_f |\nabla f|^{m-s} \nabla^{s}\varphi\big\|+\big\| |\nabla f|^{m-s} \nabla^{s}\varphi\big\|\\
\stackrel{\eqref{estimate-1}}{\lsm}&\sum_{i\leq s} \big\| |\nabla f|^{m+1-i} \nabla^{i} \varphi\big\| +\big\| |\nabla f|^{m-s} \nabla^{s}\sD_f \varphi\big\|+\big\|\varphi\big\|_{\A^{''}_{f,m}}\\
\stackrel{induction}{\lsm}&\sum_{i\leq s} \big\| |\nabla f|^{m+1-i} \nabla^{i} \varphi\big\| +\big\|\varphi\big\|_{\A_{f,m+1}}
\end{align*}
By successive application of this estimate, we have 
$$
   \big\|\varphi\big\|_{\A^{''}_{f,m+1}}\lsm \big\| |\nabla f|^{m+1} \varphi\big\|+\big\|\varphi\big\|_{\A_{f,m+1}}. 
$$
Therefore it is enough to estimate the term $\big\| |\nabla f|^{m+1} \varphi\big\|$. Note that by the $k=1$ case, 
\begin{align*}
\big\||\nabla f|^{m+1}\varphi\big\| = \big\||\nabla f| \cdot |\nabla f|^m\varphi\big\| \lsm& \big\|\sD_f |\nabla f|^m \varphi\big\| + \big\||\nabla f|^m \varphi\big\|
\end{align*}
The second term $ \big\||\nabla f|^m \varphi\big\|$ is bounded by $\big\|\varphi\big\|_{\A_{f,m}}$ by induction assumption, while the first term
\begin{align*}
\big\|\sD_f |\nabla f|^m \varphi\big\| \leq \big\||\nabla f|^m \sD_f \varphi\big\| + \big\|[\sD_f, |\nabla f|^m] \varphi\big\|.
\end{align*}
Here $[\sD_f, |\nabla f|^m]$ is the commutator of two linear operators. We have $\big\||\nabla f|^m \sD_f \varphi\big\| \lsm \big\|\sD_f \varphi\big\|_{\A_{f,m}} \leq \big\|\varphi\big\|_{\A_{f,m+1}}$ again by induction assumption. For arbitrary $\epsilon \geq 0$, we can find $N_{\epsilon}>0$ such that
\begin{align*}
\big |[\sD_f, |\nabla f|^m] \big | =& m |\nabla f|^{m-1} \big |[\sD_f, |\nabla f|] \big | \lsm m |\nabla f|^{m-1} \big |\nabla |\nabla f| \big | \\
\leq& m |\nabla f|^{m-1} |\nabla^2 f| \leq \epsilon |\nabla f|^{m+1} + N_{\epsilon} |\nabla f|^{m-1},
\end{align*}
where we have used the inequality $\big |\nabla |\nabla f| \big | \leq |\nabla^2 f|$.
We can choose $\epsilon$ small enough such that
\begin{align*}
\big\||\nabla f|^{m+1}\varphi\big\| \leq \frac{1}{2}\big\||\nabla f|^{m+1}\varphi\big\| + c \big\|\varphi\big\|_{\A_{f,m+1}},
\end{align*}
for some $c>0$. It follows that $\big\||\nabla f|^{m+1}\varphi\big\| \lsm \big\|\varphi\big\|_{\A_{f,m+1}}$. We have established the induction step 
$$
   \big\|\varphi\big\|_{\A^{''}_{f,m+1}}\lsm \big\|\varphi\big\|_{\A_{f,m+1}}. 
$$

It it now enough to prove \eqref{estimate-1} \eqref{estimate-2}.
\begin{align*}
\big\| \sD_f |\nabla f|^i \nabla^j \varphi\big\|&\leq    \big\|  [\sD_f,|\nabla f|^i \nabla^j ]\varphi\big\| + \big\|  |\nabla f|^i \nabla^j \sD_f\varphi\big\|\\
&\leq    \big\|  [\sD_f,|\nabla f|^i] \nabla^j \varphi\big\|+ \big\|  |\nabla f|^i[\sD_f, \nabla^j] \varphi\big\|  + \big\|  |\nabla f|^i \nabla^j \sD_f\varphi\big\|
\end{align*}
Again by the inequality $\big |\nabla |\nabla f| \big | \leq |\nabla^2 f|$ and strongly elliptic condition $(T)$,
$$
\big |[\sD_f,|\nabla f|^i] \big |\lsm |\nabla f|^{i+1}+ |\nabla f|^{i-1}.
$$
Therefore the first term
$$
  \big\|  [\sD_f,|\nabla f|^i ]\nabla^j \varphi\big\|\lsm \big\| |\nabla f|^{i+1}\nabla^j \varphi \big\|+\big\| |\nabla f|^{i-1}\nabla^j \varphi \big\|. 
$$
The term $[\sD_f, \nabla^j] \varphi$ is bounded by a sum of terms of the form 
$$
 \sum_{l+s\leq j} \big | |\nabla^{l+1}f| \nabla^{s}\varphi \big | \lsm  \sum_{l+s\leq j} \big | (|\nabla f|^{l+1} +1) \nabla^{s}\varphi \big |.
$$
Therefore the second term
$$
\big\|  |\nabla f|^i[\sD_f, \nabla^j] \varphi\big\|  \lsm \sum_{s\geq 0} \big\| |\nabla f|^{i+s+1} \nabla^{j-s} \varphi\big\|+ \big\| \varphi \big\|_{\A^{''}_{f,i+j}}
$$
\eqref{estimate-1} follows by combining the above estimates. 

For \eqref{estimate-2}, we start with
$$
\big\| |\nabla f|^i \nabla^{j+1} \varphi\big\|\leq \big\| [\nabla, |\nabla f|^i \nabla^{j}] \varphi\big\|+\big\| \nabla |\nabla f|^i \nabla^{j} \varphi\big\|.
$$
The rest of the estimate is completely the same as the proof of \eqref{estimate-1} above.  
\end{proof} 

\begin{proof}[Proof of Theorem \ref{equi_norm}]
Let  $\overline\A_{c,f,k}(X)$, $\overline\A'_{c,f,k}(X)$ and $\overline\A''_{c,f,k}(X)$ denote the closure of $\A_c(X)$ in $\A_{f,k}(X)$, $\A'_{f,k}(X)$ and $\A''_{f,k}(X)$ with respect to the corresponding norm.
Then we have a circle of inclusions:
\begin{align*}
\overline\A''_{c,f,k}(X) =& \overline\A'_{c,f,k}(X) = \overline\A_{c,f,k}(X) \nonumber 
= \A_{f,k}(X) \supseteq \A'_{f,k}(X) \supseteq \A''_{f,k}(X) \supseteq \overline\A''_{c,f,k}(X),
\end{align*}
where the first two equalities are due to Lemma \ref{easy} and Lemma \ref{hard}, the third equality is due to Lemma \ref{density}, the first two $\supseteq$ are due to Lemma \ref{easy} and the last $\supseteq$ is by definition.
Hence all spaces are the same and the Theorem is proved.

\end{proof}

Theorem \ref{equi_norm} is of fundamental importance for this paper. We derive a series of useful corollaries.

\begin{corollary}
Forms in $\A_{f,\infty}(X)$ are smooth, that is, $\A_{f,\infty}(X) \subset \A(X)$.
\end{corollary}
\begin{proof}
This is a direct consequence of Sobolev's embedding theorem using $\A''_{f,k}$-norm. 
\end{proof}

\begin{corollary} \label{Hodge_component}
If $\varphi=\sum\limits_{p+q=k}\varphi^{p,q} \in \A_{f,k}(X)$ where $\varphi^{p,q}$ is of Hodge type $(p,q)$. Then 
$$
\varphi^{p,q}\in \A_{f,k}(X), \quad \forall p, q. 
$$

\end{corollary}
\begin{proof}
This follows from Theorem \ref{equi_norm} and the fact that $\A''_{f,k}$-norm does not mix the Hodge types.
\end{proof}

\begin{corollary} \label{other_op}
$\A_{f,\infty}(X)$ is preserved by $\bar\pat$, $\bar\pat^*$, $df \wedge$, $(df \wedge)^*$ and their complex conjugates.
As a consequence, the Hodge decomposition given by $\pat_f$ operator also holds on $\A_{f,\infty}(X)$.
\end{corollary}
\begin{proof}
Assume $\varphi \in \A_{f,\infty}(X)$. By Corollary \ref{Hodge_component}, each Hodge component $\varphi^{p,q}$ is in $\A_{f,\infty}(X)$.
Then $\bar\pat_f \varphi^{p,q} \in \A_{f,\infty}(X)$ by Lemma \ref{trivial_op}, hence its Hodge component $\bar\pat \varphi^{p,q} \in \A_{f,\infty}(X)$. It follows that $\bar\pat \varphi \in \A_{f,\infty}(X)$.
The proofs for $\bar\pat^*$, $df \wedge$, $(df \wedge)^*$ are similar. The rest follows from the fact that the $\A''_{f,k}$-norm is invariant under complex conjugation.
\end{proof}

Another very important property of $\A_{f,\infty}(X)$ is that it admits the wedge product structure.
We prove firstly the following lemma which is similar to a result appeared in \cite{Fan}.

\begin{lemma} \label{maximum}
Assume $\varphi \in \A_{f,\infty}(X)$, then for any $k \geq 0$, $|\nabla^k \varphi|$ tends to zero as $z$ goes to infinity.
\end{lemma}
\begin{proof}
\begin{align*}
\Delta_{\bar\pat} ( \nabla^k\varphi, \nabla^k\varphi )
=& ( - g^{\bar m n} \nabla_n \nabla_{\bar m} \nabla^k\varphi, \nabla^k\varphi )
+ ( \nabla^k\varphi, - g^{\bar n m} \nabla_{\bar n} \nabla_m \nabla^k\varphi ) \\
&- g^{\bar m n} ( \nabla_n \nabla^k\varphi, \nabla_m \nabla^k\varphi )
- g^{\bar m n} ( \nabla_{\bar m} \nabla^k\varphi, \nabla_{\bar n} \nabla^k\varphi ) \\
\leq& 2 | \nabla^{k+2}\varphi| |\nabla^k\varphi| - |\nabla^{k+1} \varphi|^2.
\end{align*}

On the other hand, 
\begin{align*}
\Delta_{\bar\pat} | \nabla^k\varphi|^2
=& - 2g^{\bar m n} \nabla_n (\nabla_{\bar m} | \nabla^k\varphi| \cdot | \nabla^k\varphi|) \\
=& - 2g^{\bar m n} \nabla_n \nabla_{\bar m} | \nabla^k\varphi| \cdot | \nabla^k\varphi|
- 2g^{\bar m n} \nabla_{\bar m} | \nabla^k\varphi| \cdot \nabla_n | \nabla^k\varphi| \\
=& \Delta_d | \nabla^k\varphi| \cdot | \nabla^k\varphi| - |\nabla | \nabla^k\varphi ||^2.
\end{align*}
Using $|\nabla |\nabla^k\varphi| | \leq |\nabla \nabla^k\varphi| = |\nabla^{k+1} \varphi|$, we get the following inequality
\begin{align*}
\Delta_d |\nabla^k \varphi| \leq 2|\nabla^{k+2} \varphi|.
\end{align*}

Now we cite the Theorem 4.1 of \cite{HL} for local boundedness, which implies
\begin{align*}
\text{sup}_{B_{1/2}(z_0)} |\nabla^k\varphi| \leq c (\big\|\nabla^k\varphi\big\|_{L^2(B_1(z_0))} + \big\|\nabla^{k+2}\varphi\big\|_{L^{2n}(B_1(z_0))})
\end{align*}
where $L^{2n}(B_r(z_0))$ means the $L^{2n}$-norm in the ball $B_r(z_0)$ of radius $r$ centered at $z_0$. The constant $c$ is independent of $z_0$.
Then by Sobolev's embedding theorem, $L^{2n}(B_1(z_0))$ can be controlled by $W^{k,2}(B_2(z_0))$ when $k$ is sufficiently large, and hence by $\A_{f,k}(B_2(z_0))$. Letting $z_0$ tends to infinity, and by the fact that $c$ is independent of $z_0$, the conclusion holds.
\end{proof}

\begin{theorem} \label{form_wedge}
$\A_{f,\infty}(X)$ is closed under wedge product.
\end{theorem}
\begin{proof}
Assume $\varphi, \psi \in \A_{f,\infty}(X)$, we need to show
\begin{align*}
|\nabla f|^i \nabla^j (\varphi \wedge \psi) \in L^2_{\A}(X), \quad \forall i,j\geq 0. 
\end{align*}
By Leibniz rule, we have
\begin{align*}
|\nabla f|^i \nabla^j (\varphi \wedge \psi) = \sum_{s=0}^j \binom{j}{s} |\nabla f|^i \nabla^s \varphi \wedge \nabla^{j-s} \psi.
\end{align*}
By Theorem \ref{equi_norm}, $|\nabla f|^i \nabla^s \varphi \in L^2_{\A}(X)$, and by the previous lemma, $|\nabla^{j-s} \psi|$ is bounded. Hence each $|\nabla f|^i \nabla^s \varphi \wedge \nabla^{j-s} \psi \in L^2_{\A}(X)$ and the theorem is proved.
\end{proof}

\subsection{Comparison theorems}

In this subsection, we prove quasi-isomorphisms of several natural complexes of differential forms by generalizing a homotopy construction in \cite{LLS}. Let $\A_c(X)$ be the space of differential forms with compact support. We have natural morphisms of complexes by inclusions
\begin{equation*}
(\A_c(X), \bar\pat_f) \xlongrightarrow{i_1} (\A_{f,\infty}(X), \bar\pat_f) \xlongrightarrow{i_2} (\A(X), \bar\pat_f).
\end{equation*}

\begin{lemma} \label{inverse_grad}
Outside a neighborhood of $\text{Crit}(f)$, for any $k \geq 0$, one has
\begin{align*}
|\nabla^k (\frac{1}{|\nabla f|^2})| \lsm |\nabla f|^{k+2} + 1.
\end{align*}
\end{lemma}

\begin{proof}
We prove the lemma by induction on $k$.
The $k=0$ case holds by assumption.
For $k=1$, we have
\begin{align*}
0 = \nabla (|\nabla f|^2 \cdot \frac{1}{|\nabla f|^2}) = \nabla |\nabla f|^2 \cdot \frac{1}{|\nabla f|^2} + |\nabla f|^2 \cdot \nabla (\frac{1}{|\nabla f|^2}).
\end{align*}
Using the strongly elliptic condition \eqref{tame}, we have
\begin{align*}
|\nabla (\frac{1}{|\nabla f|^2})| = |\frac{1}{|\nabla f|^4} \cdot \nabla |\nabla f|^2|
\lsm (|\nabla f|^3 + |\nabla f|) \frac{1}{|\nabla f|^3}
\lsm |\nabla f|^3 + 1,\quad \text{as}\ z\to\infty. 
\end{align*}

Now assume the conclusion for $k < m$. Using
\begin{align*}
0 = \nabla^m (|\nabla f|^2 \cdot \frac{1}{|\nabla f|^2}) = \sum_i \binom{m}{i}\nabla^i |\nabla f|^2 \cdot \nabla^{m-i} \frac{1}{|\nabla f|^2},
\end{align*}
then the similar argument as the $k=1$ case above proves the $k=m$ case.
\end{proof}

\begin{theorem} \label{quasi_iso_stronger}
Assume $(X,g)$ has bounded geometry and $f$ satisfies the strongly elliptic condition \eqref{tame}, then both $i_1$ and $i_2$ are quasi-isomorphisms.
\end{theorem}
\begin{proof}
We only need to show both $i_1$ and $i = i_2 \circ i_1$ are quasi-isomorphisms.
To prove the case for $i$, we consider the following operator of contracting a vector field
\begin{equation*}
V_f := \frac{(df\wedge)^*}{|\nabla f|^2} = \sum_{i,j} \frac{\bar{f_i}}{|\nabla f|^2} g^{\bar i j} \iota_{\pat_j} : \quad \A(X \setminus \text{Crit}(f)) \rightarrow \A(X \setminus \text{Crit}(f)).
\end{equation*}
Direct calculation shows
\begin{equation*}
[df \wedge, V_f]=1
\end{equation*}
and
\begin{equation*}
[\bar{\pat},[\bar{\pat},V_f]]=[df \wedge,[\bar{\pat},V_f]]=[V_f,[\bar{\pat},V_f]]=0.
\end{equation*}
Let $\rho$ be a smooth function with compact support such that $\rho = 1$ in a neighborhood of $\text{Crit}(f)$.
Define another two operators on $\A(X)$ \cite{LLS}:
\begin{align*}
T_{\rho} =& \rho + (\bar{\pat} \rho)V_f \frac{1}{1+[\bar{\pat},V_f]},\\
R_{\rho} =& (1-\rho)V_f \frac{1}{1+[\bar{\pat},V_f]}.
\end{align*}
Here $\frac{1}{1+[\bar{\pat},V_f]}$ is understood as $\sum_{k\geq 0}(-1)^k [\bar{\pat},V_f]^k$, which is a finite sum by type reason. Then we have 
\begin{equation*} 
[\bar{\pat}_f,R_{\rho}]=1-T_{\rho} \quad \text{on } \A(X). 
\end{equation*}
This homotopy implies $i$ is a quasi-isomorphism.

To prove the case for $i_1$, we only need to show $[\bar{\pat}_f,R_{\rho}]=1-T_{\rho}$ holds on $\A_{f,\infty}(X)$, which amounts to show $R_{\rho}$ preserves $\A_{f,\infty}(X)$.
As $1-\rho$ vanishes in a neighborhood of $\text{Crit}(f)$, we can write $R_{\rho}(\cdot) = R_{\rho} (\eta\cdot)$, where $\eta$ is smooth function such that $\eta=0$ in a neighborhood of $\text{Crit}(f)$ and $\eta = 1$ in $X \setminus \{z| \rho(z)=1\}$.
Thus we can restrict ourselves to forms in $\A_{f,\infty}(X)$ that vanishes on $\{z| \eta(z) = 0\}$.
Denote the space of such forms by $\A^{\eta}_{f,\infty}(X)$, which is clearly preserved by $\bar\pat$ and $(df \wedge)^*$.
Using the $\A''_{f,k}$-norm and the Lemma \ref{inverse_grad}, we can show $\A^{\eta}_{f,\infty}(X)$ is also preserved by multiplication by $\frac{1}{|\nabla f|^2}$.
Now $R_{\rho}$ is a composition of $\bar\pat, (df \wedge)^*, \frac{1}{|\nabla f|^2} \cdot$ and $1-\rho$, it preserves $\A^{\eta}_{f,\infty}(X)$ and hence $\A_{f,\infty}(X)$.
\end{proof}

Now we consider the case with a formal variable $u$. We have again embeddings of complexes
\begin{align*}
(\A_c(X)((u)), Q_f) \xlongrightarrow{j_1} (\A_{f,\infty}(X)((u)), Q_f) \xlongrightarrow{j_2} (\A(X)((u)), Q_f).
\end{align*}
Recall  $Q_f=\dbar_f+u\pa$. 
\begin{theorem}
Assume $(X,g)$ has bounded geometry and $f$ satisfies the strongly elliptic condition \eqref{tame}. Then both $j_1$ and $j_2$ are quasi-isomorphisms.
Same result holds when formal Laurent series is replaced by formal power series.
\end{theorem}

\begin{proof} The proof is similar to Theorem \ref{quasi_iso_stronger}. We prove both $j_1$ and $j = j_2 \circ j_1$ are quasi-isomorphisms.
To prove the case for $j$, let $Q := \bar\pat + u\pat$ and $\rho$ be  a smooth function with compact support such that $\rho = 1$ in a neighborhood of $\text{Crit}(f)$. Define \cite{LLS}
\begin{align*}
T_{\rho}^u := \rho + [Q, \rho]V_f \frac{1}{1+[Q,V_f]} \text{ and }
R_{\rho}^u := (1-\rho) V_f \frac{1}{1+[Q,V_f]}.
\end{align*}
Then one finds
\begin{equation*} 
[Q_f, R_{\rho}^u] = 1 - T_{\rho}^u \text{ on } \A(X)((u))
\end{equation*}
which implies $j$ is a quasi-isomorphism.

As for $j_1$, one only need to prove $[Q_f, R_{\rho}^u] = 1 - T_{\rho}^u$ holds on $\A_{f,\infty}(X)((u))$, which amounts to prove $R_{\rho}^u$ preserves $\A_{f,\infty}(X)((u))$.
Expand $R_{\rho}^u$ we will get
\begin{equation*}
R_{\rho}^u = (1-\rho)V_f \sum_{i \geq 0} (-1)^i([\bar\pat, V_f] + u[\pat, V_f])^i.
\end{equation*}

Let $\A^{\eta}_{f,\infty}(X)$ be defined as in the proof of Proposition \ref{quasi_iso_stronger}, which is proved to be closed under action by $V_f$.
By Corollary \ref{other_op}, it is also closed under action by $\pat$.
Now assume $\phi(u) \in \A^{\eta}_{f,\infty}(X)((u))$, then for each $k \in \mtb{Z}$, the coefficient of $u^k$ in $R_{\rho}^u\phi(u)$ is a finite sum and each term of the sum is the output of a form in $\A^{\eta}_{f,\infty}(X)$ under the action by a finite sequence of operators in the set $\{\bar\pat, \pat, V_f\}$.
Hence the coefficient of $u^k$ in $R_{\rho}^u\phi(u)$ is in $\A^{\eta}_{f,\infty}(X)$ for each $k \in \mtb{Z}$.
We conclude that $R_{\rho}^u$ preserves $\A^{\eta}_{f,\infty}(X)((u))$.

When Laurent series is replaced by power series, the argument is the same since both $T_{\rho}^u$ and $R_{\rho}^u$ preserve power series in $u$.
\end{proof}

\subsection{Poincare duality and higher residue}

In this section,  $(X,g)$ is a bounded geometry and $f$ satisfies the strongly elliptic condition \eqref{tame}. We discuss pairings on cohomologies and dualities. 

\subsubsection*{Residue and Poincare duality}
\begin{definition}We define the following pairing on $f$-twisted spaces
\begin{align*}
 \K: \A_{f,\infty}(X) &\times \A_{-f,\infty}(X)\to \C\\
  \K(\alpha, \beta)&=\int_X \alpha\wedge \beta.
\end{align*}
\end{definition}
\noindent It is easy to see that the above integral is convergent and $\K$ is well-defined. 

\begin{proposition}\label{prop-dbar-d}The pairing $\K$ is compatible with $\dbar_f$ and $\pa$ in the sense that
\begin{align*}
\K(\dbar_f \alpha, \beta)&=-(-1)^{|\alpha|}\K(\alpha, \dbar_{-f} \beta),\\
\K(\pa \alpha, \beta)&=-(-1)^{|\alpha|}\K(\alpha, \pa \beta).
\end{align*}
Here $\alpha, \beta$ are homogenous elements of $\A_{f,\infty}(X)$ and $|\alpha|$ is the degree of $\alpha$. 
\end{proposition}

As a consequence, $\K$ induces a pairing on the cohomologies
$$
\K:  H(\A_{f,\infty}(X), \dbar_f)\times H(\A_{-f,\infty}(X), \dbar_{-f})\to \C.
$$
which we still denote by $\K$. By our comparison result Theorem \ref{quasi_iso_stronger}, we have canonical isomorphisms 
$$
H(\A_{c}(X), \dbar_f)\iso H(\A_{f,\infty}(X), \dbar_f)\iso H(\A(X), \dbar_f)
$$
\begin{definition} $\K$ induces well-defined pairings on various cohomologies
\begin{align*}
\K:  &H(\A_{c}(X), \dbar_f)\times H(\A_{c}(X), \dbar_{-f})\to \C\\
\K:  &H(\A_{f,\infty}(X), \dbar_f)\times H(\A_{-f,\infty}(X), \dbar_{-f})\to \C\\
\K:  &H(\A(X), \dbar_f)\times H(\A(X), \dbar_{-f})\to \C
\end{align*}
which we all denote by $\K$ and call Residue pairing. 
\end{definition}

\begin{remark}
When $X= \mtb{C}^n$ and $f$ has only an isolated critical point at the origin, the pairings $\K$ on cohomologies coincide with the usual residue pairing \cite{LLS}. 
\end{remark}

\begin{theorem}[Poincar\'e duality]\label{thm-duality} $H(\A_{c}(X), \dbar_f)\iso H(\A_{f,\infty}(X), \dbar_f)\iso H(\A(X), \dbar_f)$ are finite dimensional $\C$-vector spaces and the residue pairing $\K$ is non-degenerate.
\end{theorem}
\begin{proof} We prove this theorem using the model $H(\A_{f,\infty}(X), \dbar_f)$, whose elements can be represented by $\Delta_{f}$-harmonics. The finite dimensionality follows from standard elliptic analysis. 

Let $\alpha\neq 0$ be a $\Delta_f$-harmonic form and let $\beta := *\bar\alpha$. $\beta$ is $\Delta_{-f}$-harmonic by Corollary \ref{*_transf}. We have
$$
\K(\alpha, \beta) = \int_X \alpha \wedge \beta = \lge \alpha, *\bar\beta \rge = \pm \lge \alpha, \alpha \rge \neq 0
$$
since $\alpha \neq 0$. We conclude that $\K$ is non-degenerate.
\end{proof}

Duality of this type can be generalized in various ways and can be coupled to the category of matrix factorizations. See \cite{MDLM1, MDLM2,LiM} recently for some related discussions. Duality results of similar set-up also appeared in \cite{DL}.

\subsubsection*{Higher residue pairing}
Now we include a formal variable $u$ and consider the complex $(\A_{f,\infty}(X)[[u]], Q_f)$.  $\K$ is $u$-linear extended to a $\C[[u]]$-valued pairing. By Proposition \ref{prop-dbar-d}, $\K$ is compatible with $Q_f$ in the sense that 
$$
 \K(Q_f \alpha, \beta)=-(-1)^{|\alpha|}\K(\alpha, Q_{-f} \beta). 
$$
\begin{definition}\label{defn-K-forms} $\K$ induces a pairing $\hat \K$ (via $\C[[u]]$-linear extension)
$$
\hat \K: H(\A_{f,\infty}(X)[[u]], Q_f) \times H(\A_{-f,\infty}(X)[[u]], Q_{-f})\to \C[[u]]. 
$$
$\hat \K$ will be called the higher residue pairing. 
\end{definition}

Modulo $u$, the leading term of $\hat \K$ is precisely the residue pairing defined above. Higher orders in $u$ give further rich informations. As shown in \cite{LLS}, when $X= \mtb{C}^n$ and $f$ has only an isolated critical point at the origin, the pairings $\hat \K$ plays the role of K. Saito's higher residue pairing \cite{Sai3} on the formal completion of the Brieskorn lattice. Our construction of $\hat \K$ can be viewed as a generalization of higher residue pairing to Landau-Ginzburg models with compact critical locus.  In the physics literature, the $L^2$ approach to higher residue pairing was first proposed by Losev \cite{Lo}.

\section{Deformation theory}
In Part II, we discuss deformation theory on a bounded Calabi-Yau geometry with a holomorphic function satisfying the strongly elliptic condition \eqref{tame}. We construct a dGBV algebra with a trace pairing on a suitable subspace of polyvector fields, and prove the Hodge-to-de Rham degeneration via $L^2$ method. This construction unifies  Landau-Ginzburg models and compact Calabi-Yau models into the same Hodge theoretical framework. In particular, the Barannikov-Kontsevich construction of Frobenius manifolds for compact Calabi-Yau manifolds works in the same way for Landau-Ginzburg models. 

\subsection{Polyvector fields and dGBV algebra}\label{sec-PV}
Let $T_X$ denote the holomorphic tangent bundle of $X$. Let
$$
\PV^{i,j}(X) := \A^{0,j}(X, \wedge^i T_X)
$$
be the space of smooth $(0, j)$-forms valued in $\wedge^i T_X$ and
\begin{equation*}
\PV(X) := \bigoplus_{i,j} \PV^{i,j}(X).
\end{equation*}
$\PV(X)$ is bi-graded.  Elements of $\PV^{i,j}(X)$ will be called polyvector fields with Hodge degree $(i,j)$. 
In this paper, the total degree of $\PV^{i,j}(X)$ is defined to be $j-i$ and we denote by 
$$
  |\mu|=j-i, \quad \text{if}\quad \mu \in \PV^{i,j}(X). 
$$
For later use, we also let 
$$
\A_c(X)=\bigoplus_{i,j}\A^{i,j}_c(X)\subset \A(X), \quad \PV_c(X)=\bigoplus_{i,j}\PV^{i,j}_c(X)\subset \PV(X)
$$
denote subspaces consisting of elements with compact support.

Assume $X$ is a Calabi-Yau geometry equipped with a holomorphic volume form $\Omega_X$. Then $\Omega_X$ induces an isomorphism of vector spaces
\begin{align*}
\up: \PV(X) \rightarrow \A(X) \qquad
\alpha \mapsto& \alpha \vdash \Omega_X,
\end{align*}
where $\vdash \Omega_X$ denotes contraction with $\Omega_X$.
In local coordinates, 
\begin{equation*}
\bracket{d\bar z^J  {\pat_{z^I}}} \vdash \Omega_X =
(-1)^{\frac{|I|(|I|-1)}{2} }\rho d\bar z^J\wedge dz^K, \quad \text{if}\quad \Omega_X = \rho dz^I \wedge dz^K.
\end{equation*}
Here $I,J,K$ are multi-indices. For $I=\{i_1, \cdots, i_k\}$, we denote
$$
dz^I:= dz^{i_1}\wedge \cdots \wedge dz^{i_k}, \quad {\pat_{z^I}}:= {\pat_{z^{i_1}}}\wedge \cdots \wedge {\pat_{z^{i_k}}}. 
$$

Under the identification $\up$, every linear operator $P$ on $\A(X)$ induces a linear operator on $\PV(X)$ via
\begin{equation*}
\alpha\to  \up^{-1} \circ P \circ \up (\alpha), \quad \alpha\in \PV(X).
\end{equation*}
By an abuse of notation, this induced operator on $\PV(X)$ will be still denoted by $P$. Thus we have degree $1$ operators $\bar\pat, \pat, \bar\pat_f$ defined on $\PV(X)$  and $Q_f$ on $\PV(X)[[u]]$ arising from those discussed in Section \ref{sec-forms}. 

\begin{remark}
$\bar\pat$ does not depend the choice of $\Omega_X$, while $\pat$ and $Q_f$ do. Since we will fix a volume form $\Omega_X$ throughout this paper, we will not distinguish this dependence to simplify notations. 
\end{remark}
The wedge product 
$$
\PV(X)\otimes \PV(X)\to \PV(X), \quad 
\alpha \otimes \beta \longmapsto \alpha \wedge \beta
$$
equips $\PV(X)$ with a structure of graded commutative algebra. Combining $\pat$ with $\wedge$, we can define a bracket on $\PV(X)$ 
\begin{equation*}
\{\alpha, \beta\} := \pat(\alpha \wedge \beta) - \pat \alpha \wedge \beta - (-1)^{|\alpha|}\alpha \wedge \pat \beta.
\end{equation*}
This bracket does not depend on $\Omega_X$ and it coincides with the Schouten-Nijenhuis bracket up to a sign.

\begin{definition}\label{defn-BV} A dGBV algebra is a triple $(\A, d, \Delta)$ where
\begin{itemize}
\item $\A$ is a $\Z$-graded commutative associative unital algebra,
\item $\Delta: \A \to \A$ is a second-order operator of degree $1$ such that $\Delta^2=0$,
\item $d: \A\to \A$ is a derivation of degree $1$ such that $d^2=0$ and $[d, \Delta]=0$. 
\end{itemize}
\end{definition}

Here $\Delta$ is called the BV operator. $\Delta$ being ``second-order" means the following: let us define the \emph{BV bracket} $\fbracket{-,-}$ as the failure of $\Delta$ to be a derivation
$$
    \fbracket{a,b}:=\Delta(ab)-(\Delta a)b- (-1)^{\bar a}a \Delta b. 
$$
Then $\fbracket{-,-}$ defines a Lie bracket of degree $1$ (Gerstenhaber algebra) such that $\Delta$ is compatible with $\{-,-\}$ via a graded version of Leibniz rule.

The triple $(\PV(X), \bar\pat_f, \pat)$ forms  a dGBV algebra, which will be the central object of this paper.

\subsection{$L^2$ theory for polyvector fields}
 We assume $(X,g)$ has bounded geometry and $f$ is a holomorphic function satisfying the strongly elliptic condition \eqref{tame}.  $g$ induces a fiberwise hermitian product $\bracket{-,-}_{\PV}$ on polyvector fields. Let 
\begin{align*}
\varphi = \frac{1}{p!q!} {\sum_{\substack {i_1, \cdots, i_p, \\ j_1, \cdots, j_q}}} \varphi_{i_1 \cdots i_p, \bar j_1 \cdots \bar j_q} \pat_{i_1}\wedge \cdots \pat_{i_p} \otimes dz^{\bar j_1} \wedge \cdots \wedge dz^{\bar j_q}
\end{align*}
and
\begin{align*}
\psi = \frac{1}{p!q!} {\sum_{\substack {k_1, \cdots, k_p, \\ l_1, \cdots, l_q}}} \psi_{k_1 \cdots k_p, \bar l_1 \cdots \bar l_q} \pat_{k_1}\wedge \cdots \pat_{k_p} \otimes dz^{\bar l_1} \wedge \cdots \wedge dz^{\bar l_q},
\end{align*}
then the hermitian product is defined as
\begin{equation*}
(\varphi, \psi)_{\PV}(z) := \frac{1}{p!q!} \sum_{i,j,k,l} g_{i_1 \bar k_1} \cdots g_{i_p \bar k_p} g^{\bar j_1l_1} \cdots g^{\bar j_q l_q} \varphi_{i_1 \cdots i_p, \bar j_1 \cdots \bar j_q} \overline{\psi_{k_1 \cdots k_p, \bar l_1 \cdots \bar l_q}}.
\end{equation*}
This leads to an $L^2$ inner product on $\PV_c(X)$ by 
\begin{equation*}
\lge \varphi, \psi \rge_{\PV} := \int_X (\varphi, \psi)_{\PV}(z) dv_g.
\end{equation*}
 Here $dv_g$ is the volume from induced from $g$. 

In complete analogue to the discussion on differential forms, we obtain $L^2_{\PV}$ by the completion of $\PV_c(X)$ with respect to the inner product $\lge , \rge_{\PV}$. The $g$-compatible torsion free connection $\nabla$ can be used to define a connection, again denoted by $\nabla$, on the bundle $\wedge^* T_X \otimes \wedge^* \bar T^*_X$. This connection is compatible with $\bracket{-,-}_{\PV}$. Similar to Definition \ref{defn-norms} and Theorem \ref{equi_norm}, we define the following $f$-twisted Sobolev spaces of polyvector fields. 

\begin{definition}\label{defn-twisted-poly}
The spaces $\PV_{f,k}(X)$ is defined as
\begin{align*}
\PV_{f,k}(X) := \{\alpha | |\nabla f|^i \nabla^j \alpha \in L^2_{\PV}(X), \forall i+j\leq k\},
\end{align*}
and the $\PV_{f,k}$-norm is
\begin{align*}
||\alpha||_{\PV_{f,k}} := \sum_{i+j \leq k} |||\nabla f|^i \nabla^j \alpha||_{\PV}.
\end{align*}
\end{definition}

\begin{definition}\label{f-adapted-PV}
We define $\PV_{f,\infty}(X)$ to be the intersection
\begin{align*}
\PV_{f,\infty}(X) := \bigcap\limits_{k \geq 0} \PV_{f,k}(X).
\end{align*}
\end{definition}

The proof of Lemma \ref{Hodge_component} , Lemma \ref{maximum} and Theorem \ref{form_wedge} can be generalized to polyvector fields, hence we have the following:
\begin{theorem}\label{thm-PV-infinity}
Assume $(X,g)$ has bounded geometry and $f$ satisfies the strongly elliptic condition \eqref{tame}. Then $\PV_{f,\infty}(X)$ is closed under $\dbar_f$, wedge product and decomposition into components of Hodge degrees. In particular, $\PV_{f,\infty}(X)$ carries the structure of differential graded commutative algebra. 
\end{theorem}

\subsection{Bounded Calabi-Yau geometry}
Recall that the Calabi-Yau volume form $\Omega_X$ allows us to identify smooth differential forms with smooth polyvector fields through the map $\up$ as in Section \ref{sec-PV}. We have differential forms $\A_{f,\infty}(X)$ which allows the operation of $\bar\pat_f, \Delta_f$ and $\pat$, and polyvector fields $\PV_{f,\infty}(X)$ which allows the wedge product (this wedge product is different from that on differential forms under $\up$). 

We would like to compare $\up^{-1} (\A_{f,\infty}(X))$ and $\PV_{f,\infty}(X)$. In the following, we give a general sufficient condition such that these two spaces coincide. 

\begin{definition}\label{defn-bounded-CY}Let $\Theta_X$ be the holomorphic section of $\wedge^n T_X$ such that $$\Theta_X \vdash \Omega_X = 1.$$
$\Omega_X$ is called a bounded Calabi-Yau volume form with respect to $(X,g)$ if   (recall Definition \ref{defn-Sal})
$$
\Omega_X \in C^{\infty}_b(\wedge^n T^*_X)\quad \text{and}\quad  \Theta_X \in C^{\infty}_b(\wedge^n T_X).
$$ 
We define a bounded Calabi-Yau geometry to be a triple $(X, g, \Omega_X)$  where $(X,g)$ is a bounded geometry, and $\Omega_X$ is a bounded Calabi-Yau volume form.  
\end{definition}
\begin{lemma} \label{volume_condition}
Let $(X, g, \Omega_X)$ be a bounded Calabi-Yau geometry, $f$ be a holomorphic function satisfying the strongly elliptic condition \eqref{tame}. Then 
$$\up^{-1} (\A_{f,\infty}(X)) = \PV_{f,\infty}(X).$$
\end{lemma}
\begin{proof}
$\forall \alpha \in \PV_{f,\infty}(X)$ and $\forall i,j \geq 0$, we have
\begin{align*}
|\nabla f|^i \nabla^j \up(\alpha) = |\nabla f|^i \nabla^j (\alpha \vdash \Omega_X) = \sum_k \binom{j}{k} (|\nabla f|^i \nabla^k \alpha) \vdash (\nabla^{j-k} \Omega_X).
\end{align*}
By assumption, $|\nabla f|^i \nabla^k \alpha$ is $L^2$ integrable, $|\nabla^{j-k} \Omega_X|$ is bounded, so $$|\nabla f|^i \nabla^j \up(\alpha) \in L^2_{\A}(X).$$
Hence $\up(\alpha) \in \A_{f,\infty}(X)$ and $\up(PV_{f,\infty}(X)) \subset \A_{f,\infty}(X)$.
Similarly, we have $\up^{-1}(\A_{f,\infty}(X)) = \A_{f,\infty}(X) \vdash \Theta_X \subset \PV_{f,\infty}(X)$. The lemma follows.
\end{proof}

\begin{example}
For the three classes of Landau-Ginzburg models in subsection \ref{example}, one can choose a bounded Calabi-Yau volume form $\Omega_X$ as follows.
\begin{enumerate}
\item [1)] for polynomial on $\mtb{C}^n$, $\Omega_X = dz^1 \wedge \cdots \wedge dz_n$;
\item [2)] for crepant resolution on $\pi: X \rightarrow \mtb{C}^n/G$, $\Omega_X = \pi^* dz^1 \wedge \cdots \wedge dz_n$;
\item [3)] for Laurent polynomial on $(\mtb{C}^*)^n$, $\Omega_X = \frac{dz^1}{z^1} \wedge \cdots \wedge \frac{dz_n}{z^n}$.
\end{enumerate}
\end{example}

\begin{theorem}\label{thm-dGBV-property}
Let $(X, g, \Omega_X)$ be a bounded Calabi-Yau geometry, $f$ be a holomorphic function satisfying the strongly elliptic condition 
\eqref{tame}. Then $(\PV_{f,\infty}(X), \bar\pat_f, \pat)$ forms a dGBV algebra.
\end{theorem}
\begin{proof} By Theorem \ref{thm-PV-infinity}, $(\PV_{f,\infty}(X), \dbar_f)$ is a differential graded commutative algebra. By Corollary \ref{other_op} and Lemma \ref{volume_condition}, $\PV_{f,\infty}(X)$ is preserved by  $\pa$. Since $(\PV(X), \dbar_f, \pa)$ is a dGBV algebra and $\PV_{f,\infty}(X)\subset \PV(X)$, we conclude that $(\PV_{f,\infty}(X), \bar\pat_f, \pat)$ forms a dGBV subalgebra.

\end{proof}

\subsection{Hodge-to-de Rham degeneration}
Let $(\A, d, \Delta, \fbracket{-,-})$ be a dGBV algebra (either $\Z$ or $\Z/2\Z$-graded). There are two naturally associated (odd) differential graded Lie algebas:
$$
   \text{classical}: \bracket{\A, d, \fbracket{-,-}}, \quad\quad \text{quantum}: \bracket{\A[[u]], d+u\Delta, \fbracket{-,-} }. 
$$
Here in the quantum one, $u$ is a formal even variable. The quantum differential graded Lie algeba reduces to the classical one in the limit $u\to 0$. There $u$ plays the role of (formal) quantum parameter. 

Recall that the odd differential graded Lie algebra $(\A, d, \fbracket{-,-})$  is called \emph{smooth formal} if there exists a versal solution to the associated \MC equation, i.e., a degree $0$ element 
$$
\Gamma =\sum \mu_i t^i + \sum \gamma_{ij}t^it^j+ \sum \gamma_{ijk}t^it^jt^k+\cdots
$$
which satisfies 
$$
d \Gamma+{1\over 2}\fbracket{\Gamma, \Gamma}=0. 
$$
Here $\{\gamma_i\}$ is a basis of the cohomology $H(V, d)$. $\{t^i\}$ is the dual coordinate, viewed as a basis of the linear dual of $H(\A, d)$.  And $\gamma_{i_1\cdots i_k}\in V$. 

This definition extends to the quantum case. The quantum dgLa $\bracket{\A[[u]], d+u\Delta, \fbracket{-,-} }$ is called \emph{smooth formal} if there exists a degree $0$ element 
$$
\tilde \Gamma =\sum \tilde \mu_i t^i + \sum \tilde \gamma_{ij}t^it^j+ \sum \tilde \gamma_{ijk}t^it^jt^k+\cdots, \quad \tilde \gamma_{i_1 \cdots i_k}\in \A[[u]]
$$
which satisfies 
$$
(d+u \Delta) \Gamma+{1\over 2}\fbracket{\Gamma, \Gamma}=0. 
$$
Versality requires $\{\tilde \gamma_i\}$ represents a $\C[[u]]$-linear basis of $H(\A[[u]], d+u \Delta)$.  If we expand 
$
\tilde \gamma_i=\gamma_i + O(u) 
$, then the leading term $\gamma_i$ forms a basis of  $H(\A, d)$ and  $\{t^i\}$ is the dual coordinate. It is easy to see that the quantum version of smooth formal implies the classical version by taking the limit $u\to 0$. 

The smooth formality of the above quantum differential graded Lie algebra is related to the degeneration of a spectral sequence associated to the $u$-adic filtration. Precisely, by Theorem 2 of \cite{Te} (see also \cite{KKP}), the quantum differential graded Lie algebra $\bracket{\A[[u]], d+u\Delta, \fbracket{-,-} }$ is smooth formal if and only if the spectral sequence associatd to the filtration $\{F^p=u^p \A[[u]]\}$ of the complex $(\A[[u]], d+u\Delta)$ degenerates at the $E_1$-term. This amounts to saying that there exists a representative basis $\gamma_i$ of $H(\A, d)$ which extends to $\tilde \gamma_i=\gamma_i+O(u)\in \A[[u]]$ such that 
$$
   (d+u\Delta) \tilde \gamma_i=0. 
$$
It follows that $\{\tilde \gamma_i\}$ represents a $\C[[u]]$-linear basis of $H(\A[[u]], d+u\Delta)$. 

There is a vast generalization of this situation in the categorical world. The degeneration of the above spectral sequence was conjectured by Kontsevich and Soibelman \cite{KS} to hold for the Hochschild complex of  smooth and proper DG category  over a field of characteristic 0. This is proved by Kaledin \cite{Ka1, Ka2} in great generality for $\Z$-graded case. There the spectral sequence plays the role of Hodge-to-de Rham degeneration on non-commutative spaces. The categorical phase of Landau-Ginzburg models is described by matrix factorizations. Since matrix factorization DG categories are $\Z/2\Z$-graded, Kaledin's proof does not directly apply but need a variant of modification that has not been done yet \footnote{the authors would like to thank Kaledin for explaining this.}. 

Our goal in this section is to prove the  Hodge-to-de Rham degeneration on bounded Calabi-Yau geometries in the geometric phase of Landau-Ginzburg models.  We start with a useful lemma. 

\begin{lemma}\label{weak-ddbar} On  differential forms $\A_{f,\infty}(X)$, we have
$$
 \ker(\dbar_f)\cap \pa (\ker \pa_f) \subset \im (\dbar_f). 
$$
Precisely, if $\alpha \in  \ker(\dbar_f)\cap \pa (\ker \pa_f)$, then $\alpha=\bar\pat_f \bar\pat_f^* G_{\A}\alpha$. 
\end{lemma}
\begin{proof}Let us introduce the following operator
$$
W: \A_{f,\infty}(X)\to \A_{f,\infty}(X), \quad \beta=\sum_{p,q}\beta^{p,q}\to \sum_{p,q}p \beta^{p,q}. 
$$ 
Here $\beta^{p,q}$ is the $(p,q)$-form component of $\beta$. By Corollary \ref{Hodge_component}, $W$ is well-defined. Observe that 
$$
  \pa=[W, \pa_f]. 
$$
Let $\alpha \in \A_{f,\infty}(X)$ lies in $\ker(\dbar_f)\cap \pa(\ker \pa_f)$. Assume $\alpha=\pa \beta$, $\pa_f\beta=0$. Then
$$
\dbar_f \alpha=0, \quad \alpha= \pa \beta=[W, \pa_f]\beta=-\pa_f W \beta.  
$$
Recall that we have Hodge decomposition on $\A_{f,\infty}(X)$ by Corollary \ref{other_op}. The second equation implies that $\alpha$ has no harmonic component. Combining the first equation, we find
$$
\alpha=\bar\pat_f \bar\pat_f^* G_{\A}\alpha. 
$$
By Corollary \ref{other_op} and Lemma \ref{trivial_op}, $\bar\pat_f^* G_{\A}\alpha$ lies in $ \A_{f,\infty}(X)$. Hence $\alpha\in \im(\dbar_f)$. 
\end{proof}

\begin{theorem} \label{homotopy_abelian}
Let $(X, g, \Omega_X)$ be a bounded Calabi-Yau geometry and $f$ be a holomorphic function satisfying the strong elliptic condition \eqref{tame}. Let $\PV_{f,\infty}(X)$ be as in Definition \ref{f-adapted-PV}. Then the quantum differential graded Lie algebra $(\PV_{f,\infty}(X)[[u]], \dbar_f+u \pa, \fbracket{-,-})$ is smooth formal. 
\end{theorem}
\begin{proof} It is sufficient to show that the spectral sequence associated to the $u$-adic filtration of the differential complex $(\PV_{f,\infty}(X)[[u]], \dbar_f+u \pa)$ degenerates at the $E_1$-term. By Lemma \ref{volume_condition}, we can work with $\A_{f,\infty}(X)$ instead. 

For each class $[\alpha_0] \in H(\A_{f,\infty}(X), \bar\pat_f)$, we can find a $\Delta_f$-harmonic representative $\alpha_0$. We show that $\alpha_0$ can be ``corrected" to an element $\alpha= \alpha_0+ \alpha_1 u + \cdots \in \A_{f,\infty}(X)[[u]]$ such that $(\dbar_f+ u\pa)\alpha=0$. This will prove the $E_1$- degeneration. The equation $(\dbar_f+ u\pa)\alpha=0$ is equivalent to 
$$
  \pa \alpha_r=-\dbar_f \alpha_{r+1}, \quad r\geq 0.  
$$

Since $\alpha_0$ is harmonic, we have 
$$
\pa_f \alpha_0=0, \quad \dbar_f(\pa \alpha_0)=-\pa (\dbar_f\alpha_0)=0.
$$ 
Lemma \ref{weak-ddbar} applies and $\pa \alpha_0= \bar\pat_f \bar\pat_f^* G_{\A} \pat \alpha_0$. Therefore we can choose 
$$
\alpha_1=-\bar\pat_f^* G_{\A} \pat \alpha_0. 
$$

To find $\alpha_2$ next, observe that 
$$
 \dbar_f (\pa \alpha_1)=-\pa \dbar_f \alpha_1=\pa^2 \alpha_0=0, \quad \pa_f \alpha_1=-\pa_f\bar\pat_f^* G_{\A} \pat \alpha_0=-\bar\pat_f^* G_{\A} \pat  \pa_f \alpha_0=0. 
$$
Apply Lemma \ref{weak-ddbar} again, we find  $\pa \alpha_1= \bar\pat_f \bar\pat_f^* G_{\A} \pat \alpha_1$ and we can choose 
$$
\alpha_2=-\bar\pat_f^* G_{\A} \pat \alpha_1= (-\bar\pat_f^* G_{\A} \pat)^2 \alpha_0. 
$$
Inductively, by the same argument, we can solve by choosing
$$
\alpha_r= (-\bar\pat_f^* G_{\A} \pat)^r \alpha_0.
$$
\end{proof}

The proof of $E_1$ degeneration in the above theorem leads to the following:
\begin{corollary} \label{complex_splitting}
There exists a chain complex splitting
\begin{align*}
(\A_{f,\infty}(X)[[u]], Q_f) \overset{s} {\underset{t}{\leftrightarrows}} (\A_{f,\infty}(X), \bar\pat_f) \longrightarrow 0,
\end{align*}
with $t$ given by setting $u=0$ and $s$ given by
\begin{equation*}
s(\phi) = \phi + u\, \pat \bar\pat_f^*G_{\A} \phi + \sum_{i \geq 1} (-u \bar\pat_f^*G_{\A}\pat)^i P_{\A} \phi.
\end{equation*}
Recall $P_{\A}$ is the harmonic projection. 
\end{corollary}
\begin{proof} Let $\phi \in \A_{f,\infty}(X)$. By Lemma \ref{weak-ddbar}, we have
\begin{equation*}
\pat P_{\A} \phi = \bar\pat_f \bar\pat_f^* G_{\A} \pat P_{\A} \phi.
\end{equation*}
Let 
$s_1 := \pat \bar\pat_f^* G_{\A} - \bar\pat_f^*G_{\A}\pat P_{\A}$. Using the above identity 
and $P_{\A} \bar\pat_f \phi = 0$, we find
\begin{equation*}
\pat = s_1 \bar\pat_f - \bar\pat_f s_1.
\end{equation*}
Let $s_i := (-\bar\pat_f^*G_{\A}\pat)^i P_{\A}$ for $i \geq 1$, then using Lemma \ref{weak-ddbar} we find recursively
\begin{equation*}
\pat s_i = s_{i+1} \bar\pat_f - \bar\pat_f s_{i+1}.
\end{equation*}
Finally, let $s := id + \sum_{i\geq 1} u^i s_i$, we find
\begin{equation*}
Q_f \circ s = s \circ \bar\pat_f.
\end{equation*}
This proves the corollary. 
\end{proof}

\begin{corollary}\label{harmonic-spliting}
$H(\A_{f,\infty}(X)[[u]], Q_f)$ is a free $\mathbb{C}[[u]]$ module.
\end{corollary}
\begin{proof}
Let $(\phi_1, \cdots, \phi_{\mu})$ be a $\Delta_f$-harmonic basis of $H(\A_{f,\infty}(X), \bar\pat_f)$.
Since $G_{\A} \phi_k = 0$ and $P_{\A} \phi_k = \phi_k$, let
\begin{equation*}
\phi_k(u) := s(\phi_k) = \phi_k + \sum_{i \geq 1} (-u \bar\pat_f^*G_{\A}\pat)^i \phi_k.
\end{equation*}
Then $\phi_1(u), \cdots, \phi_{\mu}(u)$ generate $H(\A_{f,\infty}(X)[[u]], Q_f)$ over $\mathbb{C}[[u]]$.
\end{proof}

Now we extend our results to other familiar spaces. 
\begin{theorem}\label{lem-quasi-PV}
Let $(X, g, \Omega_X)$ be a bounded Calabi-Yau geometry and $f$ be a holomorphic function satisfying the strong elliptic condition \eqref{tame}. Then the inclusion of complexes
\begin{equation*}
(\PV_c(X), \bar\pat_f) \xlongrightarrow{i_1} (\PV_{f,\infty}(X), \bar\pat_f) \xlongrightarrow{i_2} (\PV(X), \bar\pat_f).
\end{equation*}
are quasi-isomorphisms. 
\end{theorem}
\begin{proof}This follows from Theorem \ref{quasi_iso_stronger} and Lemma \ref{volume_condition}. 

\end{proof}

\begin{theorem}\label{thm-quasi-smooth-formal}
Let $(X, g, \Omega_X)$ be a bounded Calabi-Yau geometry and $f$ be a holomorphic function satisfying the strong elliptic condition \eqref{tame}. Then the inclusion $\PV_{f,\infty}(X)\subset \PV(X)$ induces a quasi-isomorphism between two quantum differential graded Lie algebras 
$$
(\PV_{f,\infty}(X)[[u]], \dbar_f+u \pa, \fbracket{-,-}) \into (\PV(X)[[u]], \dbar_f+u \pa, \fbracket{-,-}).
$$
 In particular, $(\PV(X)[[u]], \dbar_f+u \pa, \fbracket{-,-})$ is smooth formal. 
\end{theorem}
\begin{proof}This is a direct consequence of Theorem \ref{homotopy_abelian} and Theorem \ref{lem-quasi-PV}. 

\end{proof}

\subsection{Higher residue and Frobenius manifold}
\begin{definition}We define the sesquilinear pairing 
$$
\KK_f: \PV_{f,\infty}(X)[[u]]  \times  \PV_{f,\infty}(X)[[u]] \to \C[[u]] 
$$
by
$$
\KK_f (f(u) \alpha, g(u)\beta)= f(u)g(-u) \int_X (\alpha\beta\vdash \Omega_X)\wedge \Omega_X. 
$$
\end{definition}


The following proposition is straight-forward to check. 

\begin{lemma}$\dbar_f$ is graded skew-symmetric and $\pa$ is graded symmetric with respect to the pairing $\KK_f$, i.e., 
$$
  \KK_f(\dbar_f \alpha, \beta)=-(-1)^{|\alpha|}\KK_f(\alpha, \dbar_f \beta), \quad  \KK_f(\pa \alpha, \beta)=(-1)^{|\alpha|}\KK_f(\alpha, \pa \beta), \quad \forall \alpha, \beta\in \PV_{f,\infty}(X). 
$$ 
\end{lemma}

This proposition implies that $\KK_f$ descends to cohomologies. Let us denote 
$$
\Omega_f:=H(\PV_{f,\infty}(X), \dbar_f),  \quad \hat{\mc H}_f:= H(\PV_{f,\infty}(X)((u)), Q_f), \quad \hat{\mc H}_f^{(0)}=H(\PV_{f,\infty}(X)[[u]], Q_f).
$$
By Theorem \ref{thm-duality} and Theorem \ref{homotopy_abelian}, $\hat{\mc H}_f$ is a free $\C((u))$-module and $\hat{\mc H}_f^{(0)}$ is a free $\C[[u]]$-module of the same finite rank. They are related by 
$$
\hat{\mc H}_f=\hat{\mc H}_f^{(0)}\otimes_{\C[[u]]}\C((u)), \quad \Omega_f= \hat{\mc H}_f^{(0)}/u \hat{\mc H}_f^{(0)}. 
$$ 

We can view $\hat{\mc H}_f$ as a vector bundle over the formal punctured disk $\hat \Delta^*$ parametrized by $u$, and $\hat{\mc H}_f^{(0)}$ as an extension to the origin. For an isolated singularity, $\hat{\mc H}_f^{(0)}$ is the formal completion of the associated Brieskorn lattice \cite{Sai3} (presented in the context of polyvector fields as in \cite{LLS}).

\begin{definition} $\KK_f$ defines a sesquilinear paring 
$$
\KK_f: \hat{\mc H}_f^{(0)} \times  \hat{\mc H}_f^{(0)} \to \C[[u]].
$$
We still denote it by $\KK_f$, and still call $\KK_f$ the higher residue pairing. 
\end{definition}

Next we compare $\KK_f$ on polyvector fields with $\hat \K$ on differential forms defined in Definition \ref{defn-K-forms}. 

\begin{proposition}\label{prop-compare-K} Let $(X, g, \Omega_X)$ be a bounded Calabi-Yau geometry and $f$ be a holomorphic function satisfying the strong elliptic condition \eqref{tame}. Let $\up: \PV_{f,\infty}(X)\to \A_{f,\infty}(X)$ be the contraction with  $\Omega_X$. Then
$$
  \KK_f(\alpha, \beta)=\hat \K(\up(\alpha), \widetilde{\up(\beta)}), \quad \forall \alpha, \beta \in H(\PV_{f,\infty}(X)[[u]], \dbar_f+u\pa). 
$$
Here  $\hat \K$ is in Definition \ref{defn-K-forms}. For $\beta=\sum_{k}\beta_k u^k\in \PV_{f,\infty}(X)[[u]]$, and $\beta_{k}=\sum_{i,j}\beta_k^{i,j}$ where $\beta_k^{i,j}\in \PV^{i,j}_{f, \infty}(X)$,  
$$
 \widetilde{\up(\beta)}:=\sum_{k\geq 0} \sum_{i,j=0}^n(-1)^{ni +(n+1)j} \up(\beta_k^{i,j})(-u)^k, \quad n=\dim_{\C} X,
$$
which is a well-defined cohomology class in $H(\A_{-f,\infty}(X)[[u]], Q_{-f})$.
\end{proposition}
\begin{proof} It follows from the observation that for $\alpha\in \PV_{f,\infty}^{n-i, n-j}(X), \beta\in \PV_{f,\infty}^{i,j}(X)$
$$
  \int_X (\alpha \beta\vdash \Omega_X)\wedge \Omega_X=(-1)^{nj+(n+1)i}\int_X \up(\alpha)\wedge \up(\beta). 
$$

\end{proof}

\begin{theorem}[Poincar\'e duality]\label{thm-higher-residue-duality}
Let us write $\KK_f=\sum\limits_{u\geq 0}u^k \KK_f^{(k)}$ for the higher residue pairing on $\hat{\mc H}_f^{(0)}$, and $\KK_f^{(0)}$ being the leading term. Then $\KK_f^{(0)}$ induces a pairing 
$$
\KK_f^{(0)}: \Omega_f\times \Omega_f\to \C.
$$ 
which is non-degenerate.
\end{theorem}
\begin{proof}
The non-degeneracy of this pairing follows from Theorem \ref{thm-duality} and Lemma \ref{prop-compare-K}. 
\end{proof}

\begin{remark}
In the context of isolated singularities, $\KK_f^{(0)}$ plays the role of residue pairing, and $\KK_f$ plays the role of K.Saito's higher residue pairing \cite{Sai2,LLS}. 
\end{remark}

\begin{definition}We define a $u$-connection on $\hat{\mc H}_f$ over  $\hat \Delta^*$ by 
$$
  \nabla_{\pa_u} [\alpha]:=\bbracket{\bracket{\pa_u +{1\over u}W-{f\over u^2}} \alpha}, \forall [\alpha] \in \hat{\mc H}_f.  
$$
Here $W$ is the $u$-linear extension of the following Hodge weight operator 
$$
W:\PV_{f, \infty}(X)\to \PV_{f, \infty}(X), \quad \beta=\sum_{p,q}\beta^{p,q}\to \sum_{p,q}p \beta^{p,q}, \quad \beta^{p,q}\in \PV^{p,q}_{f, \infty}(X).
$$ 
\end{definition}

It is easy to verify that as linear operators on $\PV_{f, \infty}(X)((u))$
$$
   [\nabla_{\pa_u}, Q_f]=0. 
$$
Therefore the above definition $\nabla_{\pa_u}$ is well-defined on the cohomology $\hat{\mc H}_f$. It will be more convenient to work with the logarithmic covariant derivative 
$$
   \nabla_{u\pa_u}=u\pa_u +W-{f\over u}. 
$$

The higher residue pairing $\KK_f$ is extended to $\hat{\mc H}_f$ via the same formula
\begin{align*}
   &\KK_f:  \mc H_f \times \mc H_f \to \C((u))\\
&\KK_f (f(u) \alpha, g(u)\beta)= f(u)g(-u) \int_X (\alpha\beta\vdash \Omega_X)\wedge \Omega_X. 
\end{align*}

\begin{proposition}The $u$-connection is compatible with the higher residue pairing $\KK_f$ in the following sense
$$
    (u\pa_u+n) \KK_f(\alpha, \beta)= \KK_f(\nabla_{u\pa_u}\alpha, \beta)+\KK_f(\alpha, \nabla_{u\pa_u}\beta).
$$
\end{proposition}
\begin{proof}Direct computation.
\end{proof}

 The appearance of $n$ is related to the weight of the Hodge structure. Our definition of $\KK_f$ differs from the traditional Hodge theoretical conventions in the literature \cite{Sai2} by the shift of $u^n$.

\begin{definition}  A map $\sigma: \Omega_f \to \hat{\mc H}_f^{(0)}$ is called a splitting if 
\begin{enumerate}
\item [1)] $\sigma$ is a section of the projection $\pi: \hat{\mc H}_f^{(0)}\to \Omega_f$, i.e., 
$$
\pi\circ \sigma=\text{identity}
$$
\item [2)] $\sigma$ preserves the pairing on both sides, i.e, 
$$ 
\KK_f(\sigma(\alpha), \sigma(\beta))=\KK_f^{(0)}(\alpha, \beta), \quad \forall \alpha, \beta\in \Omega_f.
$$ 
\end{enumerate}
We call $\sigma$ a good splitting if furthermore the following condition is satisfied
\begin{enumerate}
\item [3)] $\sigma(\Omega_f)[u^{-1}]$ is linear subspace of $\mc H_f$ preserved by $\nabla_{u\pa_u}$. 
\end{enumerate}
\end{definition}

A splitting is related to a $E_1$-degeneration that  splits the higher residue pairing. A good splitting requires further compatibility with the $u$-connection. In \cite{Sai2}, a good splitting is also called a good basis. The existence of good basis is a highly nontrivial problem, and is not unique in general. For quasi-homogenous isolated singularities, K. Saito constructed the good basis via the degree counting method in \cite{Sai2}, which is used to construct so-called flat structures (Frobenius manifolds in modern terminology) on the miniversal deformation space of the singularity. For general isolated singularity, the existence of good basis was proved by M. Saito \cite{SaiM} via the Hodge theory on Brieskorrn lattices. This is generalized to any tame function on a smooth affine variety \cite{B2, Sab2, Sab3, Sab4, NS, DS1,DS2}. 

In the context of $L^2$-Hodge theory, we have a natural splitting coming from harmonics.   

\begin{proposition}\label{existence-good-basis} Let 
\begin{align*}
\sigma: \Omega_f&\to \hat{\mc H}_f^{(0)}\\
       \phi & \to \phi + \sum_{i \geq 1} (-u \bar\pat_f^*G\pat)^i \phi, \quad \phi \ \text{harmonic}
\end{align*}
be the map constructed in Corollary \ref{harmonic-spliting}. Then $\sigma$ defines a splitting. 
\end{proposition}
\begin{proof}One need only to prove
\begin{align*}
\KK_f((-u\bar\pat_f^* G \pat)^i \alpha, (-u\bar\pat_f^* G \pat)^j \beta) = 0
\end{align*}
when $i+j>0$.
Without loss of generality, we assume $j>0$.  By Proposition \ref{prop-compare-K},
\begin{align*}
\KK_f((\bar\pat_f^* G \pat)^i \alpha, (\bar\pat_f^* G \pat^j \beta) 
=&  \int_X (\bar\pat_f^* G \pat)^i \up(\alpha) \wedge \widetilde{(\bar\pat_f^* G \pat)^j \up(\beta)} 
= \pm \lge (\bar\pat_f^* G \pat)^i \up (\alpha), *\overline{\widetilde{(\bar\pat_f^* G \pat)^j \up (\beta)}} \rge 
= 0.
\end{align*}
The last equality holds because it is an  $L^2$ pairing between elements in $\text{Ker}(\bar\pat_f^*)$ and $\text{Im}(\bar\pat_f)$ by Corollary \ref{*_transf}.
\end{proof}

Let $(V, Q, \Delta)$ be a dGBV algebra, $K$ be a sesquilinear paring on $V[[u]]$ with respect to which $Q$ is graded skew-symmetric and $\Delta$ is graded symmetric. Assume the quantum differential Lie algebra $(V[[u]], Q+u \Delta)$ is smooth formal, and $K$ induces a non-degenerate pairing on $H(V, Q)$. The general construction of \cite{BK} and \cite{B} gives rise to a smooth (formal) moduli space parametrized by $H(V, Q)$ based on the Bogomolov-Tian-Todorov method \cite{Bo, Ti, To}.  A further data of splitting leads to a  Frobenius manifold structure on $H(V, Q)$ . If the splitting is good, then the Frobenius manifold carries a Euler vector field. See \cite{Ma} for a review on this method.  This construction applies to $(\PV_{f,\infty}(X), \dbar_f, \pa)$ and  $\KK_f$. As a consequence, we arrive at the following theorem. 

\begin{theorem} Let $(X, g, \Omega_X)$ be a bounded Calabi-Yau geometry, $f$ be a holomorphic function satisfying the strongly elliptic condition \eqref{tame}. There exists a Frobenius manifold structure on the cohomology $H(\PV(X), \dbar_f)$. 
\end{theorem}
\begin{proof} Under the stated assumption, the inclusion $(\PV_{f,\infty}(X),\dbar_f)\subset (\PV_{f,\infty}(X), \dbar_f)$  is a quasi-isomorphism. The dGBV algebra $(\PV_{f,\infty}(X), \dbar_f, \pa)$ satisfies the $E_1$-degeneration property by Theorem \ref{homotopy_abelian}. The higher residue pairing $\KK_f$ together with a splitting by Proposition \ref{existence-good-basis} leads to a Froenius manifold structure on $H(\PV_{f,\infty}(X),\dbar_f)\iso H(\PV(X), \dbar_f)$. 
\end{proof}

Unfortunately we do not know whether the splitting in Proposition \ref{existence-good-basis} is good or not in general. This is related to the existence of Euler vector field on the corresponding Frobenius manifold. We hope to explore it in some future work.

\begin{appendix}

\section{A note on twisted de Rham cohomology}
Assume the triple $(X, f, g)$ satisfies the strongly elliptic condition \eqref{tame}.
Define a twisted de Rham operator $d_f := d + df \wedge$ on $\A(X)$.
By the identity $d + df \wedge = e^{-f} \circ d \circ e^f$, we know that the complex $(\A(X), d_f)$ compute the cohomology of $X$, which is independent of $f$. The situation is completely changed if we work with the subcomplex $(\A_{f,\infty}(X), d_f)$. We present a note on this in this appendix. 

Analogous to that for $\bar\pat_f$, we define a sequence of subspaces of $L^2_{\A}(X)$:
\begin{equation*}
\A_{d;f,k}(X) := \{\phi | (d_f + d_f^*)^i \phi \in L^2_{\A}(X) \text{ for } \forall 0 \leq i \leq k\}
\end{equation*}
and their intersection $\A_{d;f,\infty}(X) := \cap_{k \in \mtb{Z}_{\leq 0}} \A_{d;f,k}(X)$.
The corresponding $\A_{d;f,k}$-norm is defined as
\begin{equation*}
||\phi||_{\A_{d;f,k}} :=\sum_i ||(d_f + d_f^*)^i \phi||_{\A}.
\end{equation*}

\begin{theorem} \label{equi_d_f_norm}
The $\A_{d;f,k}$-norm is equivalent to the $\A_{f,k}$-norm, hence $$\A_{d;f,k}(X) = \A_{f,k}(X) \qquad \A_{d;f,\infty}(X) = \A_{f,\infty}(X).$$
\end{theorem}
\begin{proof}
One can prove the $\A_{d;f,k}$-norm is also equivalent to the $\A''_{f,k}$-norm.
We omit the details since it  is almost parallel to that for $\A_{f,k}$-norm.
\end{proof}

By Theorem \ref{equi_d_f_norm}, all the properties for $\bar\pat_f$ on $\A_{f,\infty}(X)$ hold for $d_f$ operator.
The next theorem is an analytic analog of the corollary of Theorem 1 in \cite{Sab} and Theorem 4.22 of \cite{OV}.
\begin{theorem}
We have an isomorphism of cohomologies:
\begin{equation*}
H(\A_{f,\infty}(X), \bar\pat_f) \cong H(\A_{f,\infty}(X), d_f).
\end{equation*}
\end{theorem}
\begin{proof}
The isomorphism is a combination of the following three isomorphisms.
Firstly, the map
\begin{align*}
\up_1: \A_{f,\infty}(X) \rightarrow \A_{f,\infty}(X) \quad \alpha^{p,q} \mapsto 2^{-p} \cdot \alpha^{p,q}
\end{align*}
induces an isomorphism between the complex $(\A_{f,\infty}(X), \bar\pat_f)$ and the complex $(\A_{f,\infty}(X), \bar\pat_{\frac{f}{2}})$.
 Secondly, by the K\"ahler property, we have
\begin{equation*}
H(\A_{f,\infty}(X), \bar\pat_{\frac{f}{2}}) \cong H(\A_{f,\infty}(X), d_{\text{Re} f}),
\end{equation*}
where $d_{\text{Re} f} := \bar\pat_{\frac{f}{2}} + \pat_{\frac{f}{2}} = d + d\text{Re} f \wedge$.
The isomorphism is given by identification of $\Delta_{\frac{f}{2}}$-harmonic forms.
Finally, consider the following map
\begin{align*}
\up_2: \A_{f,\infty}(X) \rightarrow \A_{f,\infty}(X) \quad \alpha \mapsto e^{i \text{Im}f} \cdot \alpha.
\end{align*}
This is a well defined isomorphism because $|e^{i \text{Im}f}| = 1$ and derivatives of $e^{i \text{Im}f}$ can be bounded by a polynomial of $|\nabla f|$ by the condition \eqref{tame}.
By the identity
\begin{equation*}
d_f = e^{-i \text{Im}f} \circ (d + d\text{Re}f \wedge) \circ e^{i \text{Im}f},
\end{equation*}
we conclude that $\up_2$ induces an isomorphism
\begin{equation*}
H(\A_{f,\infty}(X), d_f) \cong H(\A_{f,\infty}(X), d_{\text{Re} f}).
\end{equation*}
\end{proof}
\end{appendix}

\Addresses

\end{document}